\documentclass[a4paper,twocolumn,11pt,accepted=2025-12-24]{quantumarticle}
\pdfoutput=1
\usepackage[utf8]{inputenc}
\usepackage[english]{babel}
\usepackage[T1]{fontenc}
\usepackage{amsmath}
\usepackage{hyperref}

\usepackage{tikz}
\usepackage{lipsum}

\usepackage[utf8]{inputenc}
\usepackage[english]{babel}
\usepackage[T1  ]{fontenc}
\usepackage{mathtools}
\usepackage{algorithm2e}
\usepackage{alltt}
\usepackage{braket}
\usepackage{amssymb,amsthm}
\usepackage{graphicx}

\usepackage[numbers,sort&compress]{natbib}

\newcommand{\keyword}{\emph}
\newcommand{\obsSet}{\mathbb{M}}

\DeclareMathOperator{\tr}{Tr}
\DeclareMathOperator{\var}{Var}
\DeclareMathOperator{\prob}{Pr}
\DeclareMathOperator*{\argmax}{argmax}
\DeclareMathOperator*{\argmin}{argmin}
\newcommand{\jump}[1]{\mathcal{J}_{#1}}

\newtheorem{lemma}{Lemma}

\RestyleAlgo{ruled}
\SetKwInput{Kw}{input}
\SetKwData{Precompute}{Precompute}
\SetProgSty{} 

\begin{document}

\title{
Parameter estimation for quantum jump unraveling
}

\author{Marco Radaelli}
\email[]{radaellm@tcd.ie}
\affiliation{School of Physics, Trinity College Dublin, Dublin 2, Ireland}
\affiliation{Trinity Quantum Alliance, Unit 16, Trinity Technology and Enterprise Centre, Pearse Street, Dublin 2, Ireland}

\author{Joseph A. Smiga}
\email[]{joseph.smiga@rochester.edu}
\affiliation{Department of Physics and Astronomy, University of Rochester, Rochester, New York 14627, USA}

\author{Gabriel T. Landi}
\email[]{gabriel.landi@rochester.edu}
\affiliation{Department of Physics and Astronomy, University of Rochester, Rochester, New York 14627, USA}

\author{Felix C. Binder}
\email[]{felix.binder@tcd.ie}
\affiliation{School of Physics, Trinity College Dublin, Dublin 2, Ireland}
\affiliation{Trinity Quantum Alliance, Unit 16, Trinity Technology and Enterprise Centre, Pearse Street, Dublin 2, Ireland}

\newcommand{\figref}[1]{Fig.~#1}
\newcommand{\reference}[1]{Ref.~\cite{#1}}

\newcommand{\marco}[1]{\textcolor{blue}{[MR:~#1]}}
\newcommand{\gab}[1]{\textcolor{blue}{[GTL:~#1]}}

\begin{abstract}    
    We consider the estimation of parameters encoded in the measurement record of a continuously monitored quantum system in the jump unraveling, corresponding to a single-shot scenario, where information is continuously gathered. Here, it is generally difficult to assess the precision of the estimation procedure via the Fisher Information due to intricate temporal correlations and memory effects. In this paper we provide a full set of solutions to this problem. First, for multi-channel renewal processes we relate the Fisher Information to an underlying Markov chain and derive a easily computable expression for it. For non-renewal processes, we introduce a new algorithm that combines two methods: the monitoring operator method for metrology and the Gillespie algorithm which allows for efficient sampling of a stochastic form of the Fisher Information along individual quantum trajectories. We show that this stochastic Fisher Information satisfies useful properties related to estimation on a single run. Finally, we consider the case where some information is lost in data compression/post-selection and provide tools for computing the Fisher Information in this case. All scenarios are illustrated with instructive examples from quantum optics and condensed matter. 
\end{abstract}

\maketitle

In many situations of physical interest, an experimenter may want to obtain the value of a physical parameter that is not accessible via direct measurement.
To do so, they must estimate the parameter based on a finite sample of outcomes. Estimation theory~\cite{VanTrees_2004, Cover_1991, Demkowicz_2020, Kay_1993} provides techniques for retrieving the value of the parameter based on such sampling, as well as bounds on the maximum precision that can be attained.

A common approach to parameter estimation in quantum systems considers a large number of independent and identically distributed (i.i.d.) copies of the system at hand~\cite{Giovannetti_2006, Giovannetti_2011, Paris_2009}.
A considerable research effort has been devoted to determining if and how genuinely quantum properties, such as entanglement and squeezing, can yield an advantage over fully classical estimation~\cite{Pezze_2014, Pezze_2018, Braun_2018}. 

In contrast to the i.i.d. case, another scenario of experimental importance is that of a single system that is sequentially measured~\cite{Tsang2011,Gammelmark_2014,Clark_2019,Nurdin2022}. 
This can be implemented, for instance, using periodic measurements~\cite{Burgarth_2015,Yang_2023,De_Pasquale2017}, collision models~\cite{Guta_2011, vanHorssen_2015, Seah_2019, OConnor_2021,Godley2023}, or quantum continuous measurements~\cite{Mabuchi_1996, Gambetta_2001, Bouten_2004, Tsang2011, Gammelmark_2013, Gammelmark_2014, Kiilerich_2014, Berry_2015, Kiilerich_2016, Cortez_2017, Albarelli_2018, Amoros_2021, Nurdin2022, Boeyens_2023, Rinaldi_2023}.
Sometimes the system can be reset after each measurement, leading to a scenario that is akin to the i.i.d. case. However, this is not always the case. As a consequence, each data point is no longer independent, leading to a (possibly intricate) memory structure in the detection record.

In this work we provide a full characterization of quantum metrology for quantum jump unravelings~\cite{Plenio1998,Wiseman_2009}. In this scenario  a system evolves according to abrupt jumps that occur at random times and in random channels, following statistics determined by the system dynamics. 
In between jumps the evolution is smooth and governed by a modified non-Hermitian Hamiltonian. Such quantum jump unravelings are an apt description of the common experimental case in which the measurement apparatus only records the type and time of emissions from the system, e.g. in atomic physics~\cite{Kubanek_2009}, superconducting systems~\cite{Vijay_2011, Siva_2022} or cavity QED~\cite{Gleyzes_2007, Rybarczyk_2015}.

\begin{figure*}[bt]
    \centering
    \includegraphics[width=.76\textwidth]{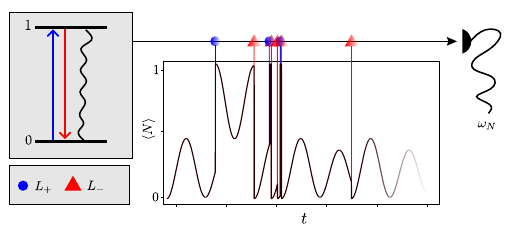}
    \caption{A quantum process emits a signal every time a jump occurs. There are two different jump channels (labeled as $L_+$ and $L_-$), corresponding to distinguishable signals. The experimenter has access only to the measurement record, constituted by jump time and jump channel for each emission. In the illustration, the quantum system is a qubit coupled to a thermal bath, whose behavior is discussed in detail in Sec.~\ref{sec:qubit_thermometry}. The plot illustrates the expectation value of the number observable for the system as a function of time.}
    \label{fig:drawing}
\end{figure*}

Note that the problem addressed here differs from metrology on the unconditional state of an open quantum system. In the latter case, the input of the estimation task is the full state of the open quantum system. Here, in contrast, the input is a measurement record as would be read on a detector during an experimental run. In this scenario, the observer has no control over the interaction between the system and the detector.

Restricted treatments of this problem were given in Refs.~\cite{Kiilerich_2014, Boeyens_2023, Gammelmark_2013} for the specific case of single-channel photodetection. 
Here we extend this to arbitrary jump processes with multiple channels. 
Fig.~\ref{fig:drawing} illustrates this scenario for the case of a qubit with two monitoring channels. 
The measurement record of such an experiment is generally of the form 
\begin{equation}
    \omega_{1:N} = \{(\tau_1, k_1), (\tau_2, k_2), \ldots, (\tau_N, k_N)\}\,,
    \label{eq:measurement_record}
\end{equation}
where $k_i$ labels the respective channels and $\tau_i$ labels the time between jumps $k_i$ and $k_{i-1}$. 
We assume that the system dynamics depend on some parameter $\theta$, which may enter via the Hamiltonian, the jump operators, or both. The challenge is then to estimate $\theta$ to high precision.
While elegant quantum bounds that maximize over all unravelings have been introduced in Refs.~\cite{Tsang2011,Gammelmark_2014,Godley2023}, these are generally not saturable in the jump unraveling. In contrast, we here focus on estimators $\hat{\theta}(\omega_{1:N})$ that only employ the (classical) measurement record in order to derive saturable precision bounds. 

We remark that since the experimenter has no control over the measurement to be performed, we do not optimize over possible measurement strategies. Consequently, our problem is one of (correlated) classical metrology, in contrast with the standard framework of quantum metrology, where the experimenter can choose the measurement to be applied~\cite{Giovannetti_2006,Giovannetti_2011}.

An interesting feature of the quantum jump scenario are the intricate memory patterns that different jumps retain of the past. 
The qubit in Fig.~\ref{fig:drawing} is an example of a renewal process, where after each jump we know with certainty which state the system jumped to. 
Multi-channel renewal systems can be mapped to semi-Markov process~\cite{Carollo2019,Brown2023} so that the memory is short-lived. 
Many quantum systems, however, are not renewal. 
A simple example is an optical cavity subject to direct photodetection: each jump perturbs the system, but does not fully reset the state. 
Quite the contrary, the memory in these systems can be extremely long lived. 
This makes it  difficult to estimate the Fisher Information (FI).
Our goal in this paper is to provide a full set of tools for doing so. 
For multi-channel renewal we show that a simple-to-use formula can be derived (Sec.~\ref{sect:Multi_channel_renewal_processes}). 
Conversely, for non-renewal processes we introduce a new algorithm based on the monitoring operator~\cite{Gammelmark_2013} and the quantum Gillespie method~\cite{Radaelli_2023b}. This approach naturally introduces the concept of a stochastic FI which can be sampled along individual trajectories by our algorithm (Sec.~\ref{sect:the_WTD_for_non_renewal_processes}). We show that the stochastic FI naturally fits into the description of single-shot estimation. A summary of these results is given in Table~\ref{tab:Fisher_information_different_systems}.

\begin{table*}[bt]
    \centering
    \begin{tabular}{ll}
        \hline\hline
        \textbf{Process type} & \textbf{Fisher information} \\
        \hline
        Single-channel renewal & $F_N(\theta) = - N \int_0^\infty d\tau W(\tau) \left(\frac{\partial^2}{\partial\theta^2} \ln W(\tau)\right)$~\cite{Kiilerich_2014} \\
        Multi-channel renewal & $F_N(\theta) = - N \sum_{m,n \in \obsSet} p_m \left(\frac{\partial^2}{\partial\theta^2} \ln W(\tau, n|m)\right)$   [Eq.~\eqref{eq:Fisher_renewal}] \\
        Multi-channel non-renewal & Numerical Gillespie algorithm [Sec.~\ref{sec:Gillespie-Fisher_monitoring}] \\ \hline\hline
    \end{tabular}
    \caption{Our work covers the calculation of the Fisher information contained in the WTD for a wide set of quantum processes, both renewal (where we extend a previously known result) and non-renewal, where we present a numerical strategy allowing for an efficient calculation.}
    \label{tab:Fisher_information_different_systems}
\end{table*}

Furthermore, as a relevant specification of our work, we discuss the loss of Fisher information under data compression/post-selection (Sec.~\ref{sect:Information_postprocessing}). 
For example, what happens if one loses the information about the distinction among channels $k_i$ in Eq.~\eqref{eq:measurement_record}? Or what happens if we build estimators based only on summary statistics, such as the average waiting times $E(\tau_i)$? 
Our results are illustrated with several examples in Sec.~\ref{sec:examples}. 

\section{Setup}

\subsection{Estimation theory}
\label{sect:estimation_theory}

Consider a string of random variables $X_{1:N}$ (here and in the following, the notation for strings is left- and right-inclusive: $X_{1:N}=X_1 X_2\ldots X_{N-1}X_N$), distributed according to a joint probability distribution $\prob_\theta(X_{1:N})$ which depends on some unknown parameter $\theta \in \Theta$. We assume $\theta$ to be a scalar; the generalization to multi-parameters is straightforward. 
The goal of estimation theory is to construct an estimator $\hat{\theta}(X_{1:N})$ that provides an estimate of $\theta$. 
The variance of an unbiased estimator ($\mathbb{E}(\hat{\theta}) = \theta$) is bounded from below by the Cram\'er-Rao bound~\cite{Fisher_1922, Cramer_1946}
\begin{equation}\label{eq:CramerRao}
    \var\left[\hat{\theta}(X_{1:N})\right] \geq \frac{1}{F_{1:N}(\theta)}\,,
\end{equation}
where the Fisher Information (FI)
\begin{equation}
    F_{1:N}(\theta) = \sum_{x_{1:N}} \prob_\theta (x_{1:N}) \big[\partial_\theta \log \prob_\theta(x_{1:N})\big]^2\,,
\label{eq:Fisher_information}
\end{equation}
represents the total information that $\prob_\theta(X_{1:N})$ contains about $\theta$. 
In the simplest scenario the $X_i$ are independent and identically distributed (i.i.d.), causing the FI to simplify as $F_{1:N}(\theta) = NF_1(\theta)$~\cite{Zegers_2015}. 
In this case, the quantity $F_1(\theta)$ represents the information rate --- i.e., the FI acquired per symbol observed. 
In the presence of correlations, in contrast, the calculation of Eq.~\eqref{eq:Fisher_information} becomes significantly involved, owing to the high dimensionality of the sum.

The Cramér-Rao bound of Eq.~\eqref{eq:CramerRao} is always asymptotically saturable, via the maximum likelihood estimation technique detailed in Sect.~\ref{sect:MLE_estimation}. In the case of finite data, however, the bound may not be saturable. If it is, the choice of the appropriate estimator may be contingent on the specific form of~$\text{Pr}_\theta$. For this reason, in the present work we focus on the computation of the Fisher information, and on the saturation of the bound in the asymptotic limit.

\subsection{The waiting time distribution}
\label{sect:the_waiting_time_distribution}

In this paper we wish to investigate the above ideas in the context of continuously measured systems in the quantum jump unraveling. 
We consider a general Markovian evolution of a quantum state $\rho(t)$, as described by a quantum master equation of the form
\begin{equation}
    \frac{d\rho}{dt} = \mathcal{L}\rho = -i[H,\rho] + \sum_{k=1}^r \mathcal{D}[L_k]\rho\,,
    \label{eq:GKSL}
\end{equation}
where $\mathcal{L}$ is the Liouvillian, $H$ is the system Hamiltonian and $\mathcal{D}$ the dissipator, defined by the relation $\mathcal{D}[L]\bullet = L\bullet L^\dagger - \frac{1}{2}\{L^\dagger L, \bullet\}$. 
$\{L_k\}$ represents a set of jump operators. Equation~\eqref{eq:GKSL} can be obtained from a closed-system microscopic description including both the system and the environment, in the limit of a weak interaction between them.
We will refer to each $k$ as a \keyword{channel}. 
In practical problems, each can represent a different physical entity. 
For example, photons might be emitted in different frequencies or polarizations. 
Conversely, in quantum dots one jump operator might represent the injection of electrons from a lead, while another represents extraction. 

Equation~\eqref{eq:GKSL} is a differential equation for the evolution of the state $\rho$. Equivalently, it can be seen as the average over stochastic trajectories, where a different effect is applied at each timestep; this approach is known as \textit{unraveling} of the master equation. The role of unravelings is two-fold: on one hand, they can be used to numerically solve the GKSL equation by sampling~\cite{Mollow_1975,Breuer_2002}. Secondly, unravelings describe what an experimenter observes in the lab under continuous monitoring. In this work, we employ unravelings in the latter sense, as single-run measurement records.

Different unravelings of Eq.~\eqref{eq:GKSL} have been proposed in the literature~\cite{Wiseman_2009}, with a prominent role being played by homodyne unraveling and quantum jump unraveling~\cite{Mollow_1975}, corresponding to different experimental setups. 

In this work, we focus on quantum jump unraveling, in which one considers quantum trajectories described by abrupt jumps occurring at random times and in random channels (as in Eq.~\eqref{eq:measurement_record})~\cite{Wiseman_2009}. 
In between jumps the system  evolves smoothly according to a (non-unitary) no-jump dynamics. 
The jump channels $k$ can thus be associated with ``detectors'' which record whenever a jump occurs in channel $k$. 
In many cases not all channels can be monitored. 
We therefore define the set of monitored channels $\obsSet$. 
We may also assume that each channel has a finite efficiency $\eta_k$ of recording a detection, such that $\eta_k = 1$ represents perfect efficiency. 
We define the superoperators
\begin{equation}
    \mathcal{J}_k \rho = \eta_k L_k \rho L_k^\dagger\,,
    \label{eq:jump_superoperators}
\end{equation}
as well as the no-jump superoperator 
\begin{equation}\label{eq:no_jump_op}
    \mathcal{L}_0 = \mathcal{L} - \sum_{k\in\obsSet} \mathcal{J}_k\,,
\end{equation}
where the sum is only over the monitored set~$\obsSet$\footnote{\label{footnote:efficiency-unmonitored} Note here that from a mathematical perspective, it is actually not necessary to use both efficiencies $\eta_k$ and a subset $\obsSet$ of monitored channels.
For example, we can treat all channels as monitored and then set $\eta_k=0$ for those channels which are actually not. 
Likewise, if we prefer not to use efficiencies, we can decompose a Lindblad dissipator as $\mathcal{D}[L_k] =  \mathcal{D}[\sqrt{\eta_k} L_k] +  \mathcal{D}[\sqrt{1-\eta_k}L_k]$ splitting each channel in two, with new jump operators $\sqrt{\eta_k} L_k$ and $\sqrt{1-\eta_k} L_k$. Then we can include $\sqrt{\eta_k} L_k$ in the subset $\obsSet$ and leave $\sqrt{1-\eta_k} L_k$ out. 
Keeping both concepts can be convenient, however, especially when making the connection with experiments.}.

The stochastic dynamics in the quantum jump unraveling can now be written as~\cite{Bouten_2004, Landi_2023} 
\begin{align}
    d\rho = dt\mathcal{L} \rho + \sum_{k\in\obsSet} &\Big(dN_k(t)- dt \tr(\mathcal{J}_k\rho)\Big) \nonumber\\ 
    & \cdot\left(
    \frac{\mathcal{J}_k \rho}{\tr(\mathcal{J}_k\rho)}
    - \rho\right)\,,    
\label{eq:quantum_jump_stochastic_master_equation}
\end{align}
where $dN_k$ is a Poisson increment --- that is, a random variable that takes the value 1 if a jump occurs in channel $k$, and 0 otherwise. The probability of this happening is $\prob(dN_k(t) = 1) = dt \tr[\mathcal{J}_k \rho]$. 
Since these probabilities are infinitesimal, at each time step $dt$ the probability of a jump taking place is vanishingly small. 
In other words, as is intuitive, most of the time $dN_k=0$ and the system will evolve with no jumps. 

Because of this structure, it is simpler (and entirely equivalent) to consider a stochastic process where we specify the time between jumps, as well as the channel in which the jump occurred. 
This is precisely the trajectory $\omega_{1:N}$ in Eq.~\eqref{eq:measurement_record} and is, in fact, the stochastic process that an experimentalist would observe. 
Thus, instead of dealing with the stochastic master equation~\eqref{eq:quantum_jump_stochastic_master_equation}, we can simply write down the joint probability density of observing the trajectory $\omega_{1:N}$. 
For concreteness, we assume that the system is initially in the state, $\rho_0$. 
The probability density of observing the measurement record Eq.~\eqref{eq:measurement_record} is then 
\begin{equation}\label{eq:trajectory_probability}
    \prob(\omega_{1:N}) = \tr\big\{ \mathcal{J}_{k_N} e^{\mathcal{L}_0 \tau_N} \cdots \mathcal{J}_{k_1} e^{\mathcal{L}_0 \tau_1} \rho_{0}\big\}\,.
\end{equation}
The question at hand is then the following: \emph{assuming that $H$ and $L_k$ can depend on an unknown parameter $\theta$, what is the precision with which we can estimate $\theta$ based only on the stochastic process in Eq.~\eqref{eq:measurement_record}?}
The answer is given by the FI, Eq.~\eqref{eq:Fisher_information}, with the stochastic string $X_{1:N}$ replaced by $\omega_{1:N}$. 
In general, Eq.~\eqref{eq:trajectory_probability}  can exhibit complex correlations: the effects of a given jump can persist, even after other jumps are measured. 
As a consequence, it cannot in general be decomposed in terms of simpler probabilities.
This makes a direct evaluation of Eq.~\eqref{eq:Fisher_information} extremely involved.
Addressing this challenge is the main contribution of this paper. 
Before proceeding, however, some comments about Eq.~\eqref{eq:trajectory_probability} are in order: 
\begin{itemize}
    \item Equation~\eqref{eq:trajectory_probability} refers only to the monitored jumps $k\in \obsSet$: in between them, any number of non-monitored jumps can occur. This is automatically taken into account via the definition of $\mathcal{L}_0$. 
    
    \item The probability given by Eq.~\eqref{eq:trajectory_probability} refers to a specific number of jumps $N$. This can be viewed as an ensemble where $N$ is fixed, but the final time $t_f = \sum_{i=1}^N \tau_i$ is allowed to fluctuate. We call this the $N$-ensemble. One could also consider an ensemble where $N$ fluctuates and $t_f$ is fixed. We will call this the $t_f$-ensemble. In the limit of large $N$ (or large $t_f$) results for the two ensembles is expected to coincide~\cite{Garrahan_2010}. 
    
    \item Equation~\eqref{eq:trajectory_probability} is not necessarily normalized. Sub-normalization occurs in cases where there is a non-zero probability that a jump never happens. Here, we exclude this possibility by assumption. That is, we work with processes where jumps must always eventually occur (called \textit{persistent} in the stochastic process literature~\cite{Marzen_2015}). Normalization is then given by $\int d\omega_{1:N} \prob(\omega_{1:N})=1$ with 
    \begin{equation}
        \int d\omega_{1:N} = \sum_{k_1,\ldots,k_N} \int\limits_0^\infty d\tau_1\cdots \int\limits_0^\infty d\tau_N\,.
    \end{equation}
\end{itemize}

\section{Multi-channel renewal processes}
\label{sect:Multi_channel_renewal_processes}

A special subclass of processes exhibiting a much-simplified memory structure are the so-called renewal processes, in which the state is reset after each jump. 
This occurs whenever the jump operators satisfy the \textit{renewal condition},
\begin{equation}
    \jump{k} \rho = \tr\left[\jump{k}\rho\right] \sigma_k\,,
    \label{eq:renewal_condition}
\end{equation}
where  $\sigma_k$ is a physical state that depends on the jump channel but not on $\rho$. 
Equation~\eqref{eq:renewal_condition} causes Eq.~\eqref{eq:trajectory_probability} to factorize as 
\begin{equation}
    \prob(\omega_{1:N}) = W(\tau_N, k_N | k_{N-1}) \cdots W(\tau_1, k_1 | k_0)\,,
\label{eq:probability_measurement_record_renewal}
\end{equation}
where
\begin{equation}
    \label{eq:wtd_renewal}
    W(\tau,k|q) = \tr\big\{ \mathcal{J}_k e^{\mathcal{L}_0 \tau} \sigma_q\big\}\,,
\end{equation}
is the waiting time distribution for the next jump to occur after time interval $\tau$ and in channel $k$, given that the previous jump occurred in channel $q$. 
Notice that this is normalized as $\sum_k \int_0^\infty W(\tau, k|q) d\tau = 1$ for all $q$. 
In terms of composite random variables 
$x_i = (\tau_i,k_i)$, it is clear that under the renewal condition $\omega_{1:N}$ itself is a Markov chain (Markov order~1 stochastic process~\cite{Ryan_2014}). 

For renewal processes in which each measurement corresponds to a single jump operator $L_k$ acting on the state, $\mathcal{J}_k \rho = \gamma_k L_k \rho L_k^\dagger$, satisfying the renewal condition [Eq.~\eqref{eq:renewal_condition}], one may show that the jump operators themselves must always have the form $L_k = \ket{\mu_k}\bra{\nu_k}$ (the constant rate $\gamma_k$ can always be chosen to satisfy this condition), where $\ket{\mu_k}$ and $\ket{\nu_k}$ are generic normalized quantum states which need not be orthogonal (c.f. Appendix~\ref{sect:form_of_the_jump_operators_for_renewal_processes} for a proof).
It then follows that the post-jump states in this renewal process are pure, $\sigma_k = \ket{\mu_k}\bra{\mu_k}$, while the jump probabilities are $\tr(\mathcal{J}_k\rho) = \gamma_k \bra{\nu_k} \rho \ket{\nu_k}$. 
The waiting time distribution in Eq.~\eqref{eq:wtd_renewal} can therefore be simplified to
\begin{equation}
    W(\tau,k|q) = \gamma_k \bra{\nu_k} e^{\mathcal{L}_0\tau}\big(\ket{\mu_q}\bra{\mu_q}\big) \ket{\nu_k}\,.
\end{equation}
If all channels are monitored, the no-jump superoperator can be written as $\mathcal{L}_0 \rho = -i (H_e \rho - \rho H_e^\dagger)$, with a non-Hermitian Hamiltonian 
\begin{equation}\label{non_hermitian_He}
    H_e = H - \frac{i}{2}\sum_k L_k^\dagger L_k = H - \frac{i}{2}\sum_k \gamma_k \ket{\nu_k}\bra{\nu_k}\,.
\end{equation}
The WTD then simplifies further to 
\begin{equation}\label{eq:renewal_W_as_probability_amplitude}
W(\tau,k|q) = |\Psi(\tau,k|q)|^2\,,    
\end{equation}
where $\Psi(\tau,k|q)$ is a ``waiting time amplitude,'' given by
\begin{equation}\label{eq:renewal_prob_amplitude}
    \Psi(\tau,k|q) = \sqrt{\gamma_k} \bra{\nu_k} e^{-i H_e \tau} \ket{\mu_q}\,.
\end{equation}

\subsection{Fisher information of multi-channel renewal processes}

We now analyze how Eq.~\eqref{eq:Fisher_information} simplifies under the factorization in Eq.~\eqref{eq:probability_measurement_record_renewal}.
As shown in Appendix~\ref{sect:Fisher_information_for_renewal_processes}, except for a small boundary term that is negligible when $N$ is large, the FI for renewal processes can be written as
\begin{equation}
    F_{1:N}(\theta) =  N \sum_{k,q \in \obsSet} p_q \int_0^\infty d\tau \frac{\big[ \partial_\theta W(\tau,k|q)\big]^2}{W(\tau, k|q)}\,.
    \label{eq:Fisher_renewal}
\end{equation}
That is, the Fisher information of the entire measurement record reduces to the Fisher information of the WTD. 
Here 
\begin{equation}\label{eq:renewal_pk}
p_k:=\frac{\tr\left[\jump{k} \rho_{\rm ss} \right]}{\sum_j\tr\left[\jump{j}\rho_{\rm ss}\right]}\,,
\end{equation}
is the steady-state probability with which jumps occur in channel $k$. 
Equation~\eqref{eq:Fisher_renewal} is our first main result. 
Notice that the scaling of the Fisher information is still linear in the number of jumps; this is expected for processes of finite Markov order~\cite{Radaelli_2023} and beyond~\cite{Riechers2023}, and is in line with Ref.~\cite{Gammelmark_2014}, which studied the quantum Fisher information maximized over all possible unravelings of the master equation. We would like to alert the reader to a possible confusion with similar expressions sometimes found in the discussion of the Fisher information for classical master equations (e.g.~\cite{Smiga_2023}): $W$ here is a function taking values over real numbers and a pair of jump labels, very different from a rate matrix.

The Cram\'er-Rao bound involving the Fisher information for a multi-channel renewal process can be asymptotically saturated by maximum likelihood estimation~\cite{VanDerVaart_2000,Cover_1991}, where the log-likelihood function takes the form
\begin{equation}\label{eq:loglikelihood_renewal}
    \ell(\tilde{\theta}|\omega_{1:N}) = \sum_{j=2}^N \log W_{\tilde{\theta}}(\tau_j,k_j|k_{j-1}).
\end{equation}
Here, $W_{\tilde{\theta}}$ represents the waiting time distribution assuming that the value of the parameter is $\tilde{\theta}$, and border terms, that do not play any role in the asymptotic limit, are neglected. Other estimators generally perform worse than the MLE, at least in the limit of a large number of jumps on the trajectory. An example of an estimation task performed on the qubit thermometry system (as defined in Sec.~\ref{sec:qubit_thermometry}) is discussed in Appendix~\ref{app:MLE_qubit_thermo}.

In the case of a single jump channel Eq.~\eqref{eq:probability_measurement_record_renewal} will factor into a product of probabilities. 
The waiting times $\tau_i$ therefore become independent and identically distributed (the process becomes of Markov order~0). This is the scenario of renewal theory~\cite{Vacchini_2020} and Eq.~\eqref{eq:Fisher_renewal}  simplifies to 
\begin{equation}\label{eq:Fisher_renewal_single_channel_iid}
    F_{1:N}(\theta) = N \int\limits_0^\infty d\tau\,
    \frac{\big[\partial_\theta W(\tau)\big]^2}{W(\tau)}\,,
\end{equation}
where $W(\tau) = \tr\big\{\mathcal{J} e^{\mathcal{L}_0 \tau} \sigma\big\}$. 
This agrees with a result derived in Ref.~\cite{Gammelmark_2014}. 

The Fisher information in Eq.~\eqref{eq:Fisher_renewal} encompasses the information contained in both the jump channels and the jump times. 
It is possible to separate the two individual contributions using the basic decomposition law $F(X,Y) = F(Y) + F(X|Y)$, which holds for the Fisher information of any joint probability distribution~\cite{Zegers_2015,Micadei2015,Radaelli_2023}. As a result, one finds
\begin{equation}\label{eq:Fisher_renewal_splitting}
    F_{1:N}(\theta) = F_{1:N}^{\rm ch}(\theta) + F_{1:N}^{\rm times|ch}(\theta)\,,
\end{equation}
where 
\begin{equation}\label{eq:Fisher_renewal_channels}
    F_{1:N}^{\rm ch}(\theta) = N \sum_{k,q\in\obsSet} p(k|q)p_q \big(\partial_\theta \ln p(k|q)\big)^2\,,
\end{equation}
is the Fisher information of the sequence of jump channels only, with 
\begin{equation}\label{eq:renewal_markov_transition}
    p(k|q) = \int\limits_0^\infty d\tau~W(\tau,k|q) = - \tr\big\{ \mathcal{J}_k \mathcal{L}_0^{-1} \sigma_q\big\}\,,
\end{equation}
being the probability distribution that the system transitions from $q\to k$, irrespective of when it happens. Here, we assume that the inverse of the no-jump Lindbladian exists; this is a sufficient condition for the absence of dark subspaces.

We mention in passing that the quantity $p_k$ in Eq.~\eqref{eq:renewal_pk} is also the solution of 
\begin{equation}\label{eq:renewal_markov_chain_jump_channels}
    \sum_q p(k|q)p_q = p_k\,,
\end{equation}
as one may verify by explicit substitution. 
The second term in Eq.~\eqref{eq:Fisher_renewal_splitting}, on the other hand, is the conditional Fisher information of the time tags, given one knows the sequence of channels:
\begin{equation}\label{eq:Fisher_renewal_times_given_channels}
\begin{aligned}
    F_{1:N}^{\rm times|ch}(\theta) = N\sum_{k,q\in\obsSet} p(k|q) p_q 
    \int\limits_0^\infty d\tau\,\frac{\big[ \partial_\theta W(\tau|k,q)\big]^2}{W(\tau|k,q)}\,,
\end{aligned}
\end{equation}
where 
\begin{equation}\label{eq:renewal_splitting_waiting_times_conditioning}
    W(\tau|k,q) = \frac{W(\tau,k|q)}{p(k|q)}\,,
\end{equation}
is the probability that the jump takes a time $\tau$, given that the sequence of jumps was $q\to k$. 
The quantity $F_{1:N}^{\rm times|ch}$ is therefore the FI of the time tags, given that one knows the channels. 
The FI of the time tags when one is ignorant about the channels will be discussed in Sec.~\ref{sect:Information_postprocessing}.

The two terms in Eq.~\eqref{eq:Fisher_renewal_splitting} are individually non-negative, and have clear physical interpretations. Therefore, they can be used to assess how much information is contained in each aspect of the stochastic process (channels and times). 
The relative contribution of $F^{\rm ch}$ to the total Fisher information is illustrated in Fig.~\ref{fig:FI_qubit_thermometry}(b-c).

In passing, we mention that if 
we insert Eq.~\eqref{eq:renewal_splitting_waiting_times_conditioning} into the general probability string Eq.~\eqref{eq:probability_measurement_record_renewal}, we can write it as
\begin{equation}
\begin{aligned}
    \prob(\omega_{1:N}) &= \Big[W(\tau_N|k_N,k_{N-1}) \cdots W(\tau_1|k_1,k_0)\Big] \\
     &\qquad \times \Big[p(k_N|k_{N-1}) \cdots p(k_1|k_0)\Big]\,.
\end{aligned}    
\end{equation}
Thus, we see that for renewal processes the path trajectory can be split into a probability for the sequence of channels $k_0\to k_1\to \cdots$ times the conditional probability of the jump times $\tau_i$, given this sequence of channels.

The integral appearing in Eq.~\eqref{eq:Fisher_renewal} can be difficult to compute analytically, depending on the shape of the WTD. The same is true for Eq.~\eqref{eq:Fisher_renewal_times_given_channels}. On the other hand, Eqs.~\eqref{eq:renewal_markov_transition} and~\eqref{eq:Fisher_renewal_splitting} allow for an easy computation of the Fisher information when time tags are ignored. In Appendix~\ref{sect:particular_cases}, we discuss the special case of a WTD that is a single decaying exponential, for which the integration can be analytically carried out; this case most frequently appears in incoherent systems described by classical master equations. In Appendix~\ref{app:FI_bound}, we introduce a bound on the Fisher information in terms of the mean and variance of the WTD, as
\begin{equation}
    F_{1:N}(\theta) \geq F_{1:N}^{\rm ch}(\theta) + N \sum_{k,q} \frac{\left[\partial_\theta \mu_{kq}\right]^2}{\sigma_{kq}^2} p\left( k\middle\vert q \right)p_q\,,
\end{equation}
where $\mu_{kq}=\mathbb{E}\left[ \tau \middle\vert k,q \right]$ and $\sigma_{kq}^2=\var\left[ \tau \middle\vert k,q \right]$ are the conditional mean and variance of $\tau$, calculated from Eq.~\eqref{eq:renewal_splitting_waiting_times_conditioning}. This bound can be useful in situations in which the integral of Eq.~\ref{eq:Fisher_renewal} is hard to compute, but the variance and the mean of the WTD is computable. 

One may also consider the inverse decomposition of the Fisher information between channels and jumps, of the form $F_{1:N}(\theta) = F_{1:N}^{\text{times}}(\theta) + F_{1:N}^{\text{ch}\vert\text{times}}(\theta)$. While this decomposition is formally valid, the times-only process is, in general, non-renewal even though the original process is renewal. This prevents us from expressing~$F_{1:N}^{\text{times}}$ analytically. The computation can be performed numerically, as detailed in Sect.~\ref{sect:Information_postprocessing}.

\section{Non-renewal processes}
\label{sect:the_WTD_for_non_renewal_processes}

The calculation of the FI in Eq.~\eqref{eq:Fisher_information} for a general measurement record $\omega_{1:N}$ is difficult because the probability does not factorize in any simple way as it did for renewal processes in Eq.~\eqref{eq:probability_measurement_record_renewal}. 
As a consequence, one must compute a large number of high-dimensional sums/integrals, rendering the problem generally intractable. 
In this section we start by reviewing the monitoring operator formalism~\cite{Gammelmark_2013}, which allows one to compute the Fisher information by sampling over quantum trajectories.
We then adapt this formalism to our specific ensemble distribution given by Eq.~\eqref{eq:trajectory_probability}, and present a new algorithm for efficient computation of the Fisher information over non-renewal quantum trajectories, based on the quantum Gillespie algorithm introduced in Ref.~\cite{Radaelli_2023b}.

\subsection{Monitoring operator for generic stochastic quantum mechanics} \label{sec:Gillespie-Fisher_monitoring}
Consider a generic set of trace-non-increasing superoperators $\{\mathcal{M}_x\}$ that describe the evolution of the system after outcome $x$. Namely, if the system is initially in the normalized state $\rho$, then
\begin{equation}
    \tilde{\rho}' = \mathcal{M}_x \rho
\end{equation}
is the unnormalized state such that $\tr \tilde{\rho}'$ is the probability of measuring outcome $x$ and $\rho'=\frac{\tilde{\rho}'}{\tr \tilde{\rho}'}$ is the corresponding normalized state. Stringing many measurements together, starting from an initial state $\rho_0$, one defines the unnormalized density operator
\begin{equation}
    \tilde{\rho}_N = \mathcal{M}_{x_N} \cdots \mathcal{M}_{x_1} \rho_0\,.
\end{equation}
The corresponding probability is 
\begin{equation}\label{eq:prob_Mx_gen}
    \prob(x_1,\ldots,x_N) = \tr\{\tilde{\rho}_N\}. 
\end{equation}
For the case of quantum jumps, $\mathcal{M}_x = \mathcal{J}_k e^{\mathcal{L}_0\tau}$, so the outcomes are given by the pair $x = (k,\tau)$.

The \keyword{monitoring operator} for an unknown parameter $\theta$ is defined as as~\cite{Gammelmark_2013} 
\begin{equation}
    \xi_{N} = \frac{\partial_\theta \tilde{\rho}_{N}}{\tr\{\tilde{\rho}_{N}\}}\,.
    \label{eq:monitoring_operator}
\end{equation}
One may verify that 
\begin{equation}
        F_{1:N}(\theta) = \mathbb{E}\big[ \tr(\xi_{N})^2 \big]\,,
    \label{eq:FI_monitoring_op}
\end{equation}
is precisely the general FI in Eq.~\eqref{eq:Fisher_information} (here the ensemble average is over Eq.~\eqref{eq:prob_Mx_gen}).
Hence, stochastically sampling the monitoring operator provides us a method for computing the Fisher information in the non-renewal case. 
In addition, as we also discuss in Sec.~\ref{sect:MLE_estimation}, 
the monitoring operator can be used to perform maximum likelihood estimation tasks on quantum jump trajectories, solving many numerical instabilities issues with respect to a more traditional approach.

A major benefit of the monitoring operator formalism is that there is a simple, inductive process to calculate it. First, for an initial state $\rho_0$, the monitoring operator is initialized as 
\begin{equation}\label{eq:xi_0}
    \xi_0 = \partial_\theta \rho_0\,.
\end{equation}
Starting from the state, $\rho_j$, the next observation $x_{j+1}$ is sampled from
\begin{equation}
    \prob(x_{j+1}|\rho_j) = \tr\{\mathcal{M}_{x_{j+1}} \rho_j\}\,,
\end{equation}
and the state evolves to $\rho_{j+1}=\frac{\mathcal{M}_{x_{j+1}} \rho_j}{\tr\{\mathcal{M}_{x_{j+1}} \rho_j\}}$. Meanwhile, the monitoring operator evolves to
\begin{equation}\label{eq:generic_MO_step}
    \xi_{j+1} = \frac{(\partial_\theta \mathcal{M}_{x_{j+1}})\rho_j + \mathcal{M}_{x_{j+1}}\xi_j}{\tr\{\mathcal{M}_{x_{j+1}}\rho_{j}\}}\,.
\end{equation}
Note that this calculation only depends on the previous monitoring operator $\xi_j$, the superoperators $\{\mathcal{M}_x\}$, and the normalized density operator $\rho_j$.
This process does not depend on keeping track of the unnormalized state, which leads to better numeric stability~\cite{Albarelli_2018}.

The inductive algorithm described above enables a computationally simple way to determine the evolution of the state $\rho_j$ and monitoring operator $\xi_j$ from their respective initial conditions. For calculating the Fisher information, this process is repeated many times to calculate $F_{1:N}(\theta) = \mathbb{E}\big[ \tr(\xi_{N})^2 \big]$. In the next section, we show how to use the monitoring operator formalism to perform MLE estimation on general quantum stochastic processes.

We remark that the monitoring operator formalism is fully general with respect to the specific unraveling of a given master equation. For instance, it can be applied to homodyne unraveling~\cite{Gammelmark_2013, Albarelli_2018}. Consequently, most of the results presented below can be adapted to homodyne detection or other protocols, with the exception of the use of the Fisher-Gillespie algorithm, which crucially relies on the jump structure of the evolution.

\subsection{MLE estimation}
\label{sect:MLE_estimation}
Consider a measurement record $\omega_{1:N}$, i.e. a string of outcomes $(x_1,\ldots,x_N)$. This is the outcome of a single run of the experiment. The \textit{likelihood} is the probability of obtaining $\omega_{1:N}$, as a function $\theta$, the parameter to be estimated:
\begin{equation}
    \ell(\tilde{\theta}) = \prob\left[\omega_{1:N}\vert\tilde{\theta}\right]\,.
\end{equation}
In MLE~\cite{Kay_1993, VanDerVaart_2000} the estimator for $\theta$ on a single run is taken as the maximum of the $\ell(\tilde{\theta})$ function over all values of $\tilde{\theta}$:
\begin{equation}
    \hat{\theta} = \argmax_{\tilde{\theta}} \ell(\tilde{\theta})\,.
\end{equation}
In other words, MLE assumes that the currently observed measurement record must come from the value of the parameter that makes it the most likely compared to other values of $\theta$. It is possible to prove that, under suitable regularity conditions, MLE is asymptotically unbiased and saturates the Cram\'er-Rao bound. 

In order to be able to perform the maximization, one has to be able to compute the likelihood $\ell(\tilde{\theta})$ for many different values of $\tilde{\theta}$. When considering long evolutions (that is, most situations of practical relevance in quantum metrology), we are faced with the problem that $\ell(\tilde{\theta})$ decays exponentially with the length of the measurement record. This renders the estimation task numerically unfeasible. In the context of finitely correlated processes (i.e., processes with a finite Markov order), the issue can be successfully tackled by maximizing the logarithm of the likelihood function, and exploiting the factorization properties of logarithms~\cite{Radaelli_2023}. However, quantum processes in general exhibit an infinite Markov order. Here, we show a method, based on the monitoring operator formalism, that circumvents this issue. An alternative but equivalent method has been proposed in Ref.~\cite{Gammelmark_2013}.

Given the measurement record $\omega_{1:N}$, and a candidate value for the parameter $\tilde{\theta}$, we can define a corresponding monitoring operator $\xi_{\omega_{1:N}}^{\tilde{\theta}}$, that is, the monitoring operator evolved along the measurement record $\omega_{1:N}$, assuming that the value of the parameter is $\tilde{\theta}$. In an analogous way, we can evolve the sub-normalized quantum state $\tilde{\rho}_{\omega_{1:N}}^{\tilde{\theta}}$ on $\omega_{1:N}$. The likelihood can then be written as
\begin{equation}
    \ell\left(\tilde{\theta}\right) = \prob\left[\omega_{1:N}\vert\tilde{\theta}\right] = \tr\left[\tilde{\rho}_{\omega_{1:N}}^{\tilde{\theta}}\right]\,.
\end{equation}
Let us now assume that, in the estimation region, the loglikelihood function has a single maximum. Note that for well-behaved loglikelihood functions this is not a strong constraint since one can always restrict the estimation region. If the maximum lies in the interior of the region, then it can be detected by imposing $\partial_\theta \ell(\theta)=0$. If the maximum lies on the border, it cannot be detected in this way. For this reason, it makes sense to check the loglikelihood function on the border separately, for instance with the method of Ref.~\cite{Gammelmark_2013}. We now focus on the case of a maximum inside the region.

The maximization condition yields
\begin{equation}
    \partial_{\tilde{\theta}} \tr\left[\tilde{\rho}_{\omega_{1:N}}^{\tilde{\theta}}\right] =0\iff \tr\left[\xi_{\omega_{1:N}}^{\tilde{\theta}}\right] =0\,.
\end{equation}
In this way, we mapped the problem of finding the $\tilde{\theta}$ that maximizes the likelihood to the one of finding the $\tilde{\theta}$ that makes the monitoring operator traceless. Remarkably, the latter is a numerically stable problem, because the evolution of the monitoring operator can be integrally computed from the normalized conditional state. For practical purposes, the problem can be further mapped onto finding the $\tilde{\theta}$ that minimizes the squared trace of the monitoring operator. To summarize:
\begin{equation}
    \hat{\theta} = \argmax_{\tilde{\theta}} \ell(\tilde{\theta}) = \argmin_{\tilde{\theta}} \left(\tr\left[\xi_{\omega{1:N}}^{\tilde{\theta}}\right]\right)^2\,.
\end{equation} 

\begin{figure*}[bt]
    \includegraphics[width=\textwidth]{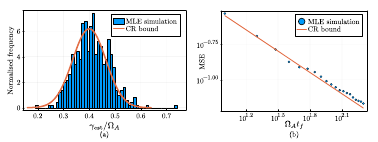}
    \caption{Estimation task results for the coupled qubit model. MLE estimation on 500 trajectories for a fixed final time~100 (a) yields normally distributed estimates; the variance is compatible with the prediction of the Cram\'er-Rao bound. (b): the variance of the MLE estimates is compared with the Cram\'er-Rao bound. As expected, perfect saturation is not reached for finite time, but the error on the estimates is comparable in order of magnitude with the bound. Parameters: $\Omega_A = \Omega_B = \bar{n} = 1$, $\gamma = 0.4$, $g = 0.01$, timestep $dt = 0.001$, $\gamma$ increment for the derivatives $d\gamma = 0.0001$.}
    \label{fig:MLE_CR_bound}
\end{figure*}

We remark that the monitoring operator formalism for MLE estimation is fully general, regardless of the specific form of the $\mathcal{M}_x$ superoperators. An application to the case of quantum jump trajectories is presented in Fig.~\ref{fig:MLE_CR_bound} for the coupled qubit model, whose details are discussed in Sec.~\ref{sect:coupled_qubits}. In the asymptotic limit of a large final time~$t_{f}$, MLE gives normally distributed estimates for the unknown parameter, and the variance of such distribution saturates the Cram\'er-Rao bound. Code for MLE estimation on quantum jump trajectories, using the monitoring operator formalism, is made available in Julia~\cite{implementation_Julia}. 

\subsection{The Fisher-Gillespie algorithm}
To compute the Fisher information according to Eq.~\eqref{eq:FI_monitoring_op}, one has to sample from the distribution of all the possible sequences of superoperators $\mathcal{M}_x$, weighted by the appropriate probability of the trajectory. In the case of quantum jumps, $\mathcal{M}_x = \mathcal{J}_k e^{\mathcal{L}_0\tau}$, so the time evolutions add an additional level of complexity as the superoperators are parameterized by a continuous variable, time, that can be arbitrarily subdivided.

To make this computation tractable, we introduce a new algorithm, which jointly exploits the monitoring operator formalism for the Fisher information, and the Gillespie approach~\cite{Gillespie_1976,Gillespie_1977} to quantum jump trajectories, introduced in Ref.~\cite{Radaelli_2023b}. Here, we give a quick overview of the working principles of the algorithm.

To sample from the distribution of all possible jump trajectories given by Eq.~\eqref{eq:trajectory_probability}, we assume that we start in a state $\rho_{j-1}$, immediately after a jump. We then sample a random waiting time $T_{j}$ from the distribution
\begin{equation}
    W(\tau|\rho_{j-1}) = \tr\left[\mathcal{J}e^{\mathcal{L}_0\tau} \rho_{j-1}\right]\,, 
    \label{eq:algorithm_WTD}
\end{equation}
where $\mathcal{J} = \sum_{k\in\obsSet} \mathcal{J}_k$. This gives the WTD for the next jump to occur, irrespective of the monitored channel it takes place in. 
We then update the state $\rho_{j-1}$ to 
\begin{equation}\label{eq:Gillespie_evolve_state}
    \bar{\rho}_j = \frac{e^{\mathcal{L}_0 T_j} \rho_{j-1}}{\tr(e^{\mathcal{L}_0 T_j} \rho_{j-1})}\,.
\end{equation}
Next we sample the jump channel $k_j$ from the probabilities 
\begin{equation}
    p(k|\bar{\rho}_j) = \frac{\tr(\mathcal{J}_{k}\bar{\rho}_j)}{\sum_q \tr(\mathcal{J}_{q}\bar{\rho}_j)}\,,
\end{equation}
and use the outcome $k_j$ to update the state to 
\begin{equation}
    \rho_j = \frac{\mathcal{J}_{k_j} \bar{\rho_j}}{\tr(\mathcal{J}_{k_j} \bar{\rho_j})}\,.
\end{equation}
We now continue in this way. At each step, the time $T_j$ is sampled from $W(\tau|\rho_{j-1})$ and the next jump $k_j$ is sampled from $p(k|\bar{\rho}_j)$. Combining these steps, the state is updated to 
\begin{equation}
    \rho_j = \frac{\mathcal{J}_{k_j}e^{\mathcal{L}_0 T_j} \rho_{j-1}}{\tr(\mathcal{J}_{k_j} e^{\mathcal{L}_0 T_j} \rho_{j-1})}\,.
    \label{eq:rho_dynamics_step}
\end{equation}
This gives a way to efficiently obtain the correct weight~\eqref{eq:trajectory_probability} for each trajectory. Now we have to consider the evolution of the monitoring operator on the same trajectory, according to Eq.~\eqref{eq:generic_MO_step}, where $\mathcal{M}_{x_j}=\mathcal{J}_{k_j} e^{\mathcal{L}_0 T_j}$,
\begin{equation}
    \xi_j = \frac{\partial_\theta\big(\mathcal{J}_{k_j} e^{\mathcal{L}_0 T_j}\big) \rho_{j-1} + \mathcal{J}_{k_j} e^{\mathcal{L}_0 T_j} \xi_{j-1}
    }{\tr(\mathcal{J}_{k_j} e^{\mathcal{L}_0 T_j} \rho_{j-1})}\,.
    \label{eq:monitoring_operator_dynamics_step}
\end{equation}
This equation can be read in direct comparison with the results of Ref.~\cite{Albarelli_2018}. In our case, however, the evolution is written in the Gillespie formalism, rather than the time-discretized picture. The FI is then obtained by repeating this process over sufficiently many trajectories to get the average~\eqref{eq:FI_monitoring_op}. 

Importantly, many of the numerically heavier objects involved do not depend on the specific trajectory and can be pre-computed. To do so, in analogy with the Gillespie algorithm of Ref.~\cite{Radaelli_2023b}, we define a list $\{T_\alpha\}$ of times, which has to (i)~be sufficiently fine-grained to be able to resolve all relevant features of $W(\tau|\rho)$ and (ii)~span times much larger than the average waiting times, in order to cover as faithfully as possible the structure of the WTD. At this point, we can pre-compute not only the superoperators $e^{\mathcal{L}_0\tau}$ required to evolve the state in Eq.~\eqref{eq:Gillespie_evolve_state}, but also the superoperators $\partial_\theta \mathcal{J}_{k}$ for all the jump channels and $\partial_\theta e^{\mathcal{L}_0 \tau}$, needed to evolve the monitoring operator $\xi_t$ according to Eq.~\eqref{eq:monitoring_operator_dynamics_step}.

If the initial state is pure, all jump operators are one-dimensional projectors, and all jump channels are monitored, then the conditional state remains a pure state for the entire duration of the dynamics. It is therefore possible to write its evolution in terms of the effective Hamiltonian $H_e = H - i/2 \sum_k \mathcal{J}_k$, with significant savings in terms of memory and computational time.

A pseudocode description of the Gillespie-Fisher algorithm in this simpler case is given in Algorithm~\ref{alg:Gillespie}; an implementation in Julia can be found in Ref.~\cite{implementation_Julia}. The latter allows both for pure state evolution and for mixed states, making also possible to consider the partial monitoring of specific channels.

At the beginning of the dynamics, $\xi_0 = \partial_\theta \rho_0$ [Eq.~\eqref{eq:xi_0}]; it follows that the monitoring operator is initially traceless (hence $F_0(\theta)=0$), in accordance with the fact that no information is available about $\theta$ at the beginning. However, the initial state may depend on $\theta$; e.g., one could be given the steady state of the system as the initial state, and the steady state ordinarily depends on the dynamical parameters. While this cannot play a role in the initial Fisher information, it may become relevant at later times, when the value of $\xi$ is modified by the subsequent application of the jump and no-jump superoperators. In principle henceforth all that follows may depend upon the specific choice of the initial state $\rho_0$. However, in cases wherein the information about the initial state is quickly lost during the time evolution, these effects eventually become vanishingly small.

Finally, we observe that the Fisher-Gillespie algorithm is, in effect, a Monte Carlo integration method for the Fisher information. By following the evolution, it allows to focus on the most relevant areas of the space of measurement records. Since the trajectories are generated independently of each other, the error in the evaluation of the Fisher information follows from the law of large numbers, scaling as~$1/\sqrt{M}$, where $M$ is the number of trajectories.

\subsection{Fisher information rate}
In this section, we discuss the Fisher information rate $dF/dt$ of the measurement record of any quantum jump process. For this, we need to work in the $t_f$-ensemble representation: the measurement record $\omega_{1:N}$ is equivalently represented by a sequence of discrete time steps of equal length $dt$. For each step $t$, a random variable $X_t$ can take values in the set $\{0,1,\ldots, |\obsSet|\}$, where the value 0 represents the absence of a jump in the step, and the other values index the possible jump channels. 

Given the Fisher information in the measurement record up to time step $T$, $F_\theta(X_{1:T})$, one can update it to the next time step $T+1$ via the chain rule for the Fisher information~\cite{Zegers_2015, Radaelli_2023}
\begin{equation}
    F_\theta(X_{1:T+1}) = F_\theta(X_{1:T}) + F_\theta(X_{T+1}|X_{1:T})\,,
\end{equation}
where $F(X|Y)$ represents the conditional Fisher information for $X$ given $Y$. The second term gives the variation of the Fisher information $dF_T$. By definition of the conditional Fisher information:
\begin{equation}
\begin{split}
    dF_T = & F(X_{T+1}|X_{1:T}) \\= & \sum_{x_{1:T+1}} \frac{\prob(x_{1:T})\left(\partial_\theta \prob(x_{T+1}|x_{1:T})\right)^2}{\prob(x_{T+1}|x_{1:T})}\,.
\end{split}
\end{equation}
The update probability $\prob(x_{T+1}=j|x_{1:T})$, for $j=1,\ldots,|\obsSet|$, is given by
\begin{equation}
\begin{split}
    \prob(x_{T+1}=j|x_{1:T}) = I^j_T dt\,,
    \label{eq:Fisher_rate}
\end{split}
\end{equation}
and $\prob(x_{T+1}=0|x_{1:T}) = 1 - \sum_{j=1}^{|\mathbb{M}|} I_{T}^{j} dt$, where $I_T^j\equiv \tr\left[L_j\rho^{(c)}_T L_j^\dagger\right]$ is the stochastic current in channel $j$ with corresponding jump operator $L_j$ and $\rho_T^{(c)}$ the conditional normalized state at time $T$
 Expanding the calculations (see Appendix~\ref{sect:FI_rate_average}), we find the Fisher information rate to be
\begin{equation}
    \frac{dF_T}{dt} = \sum_{k=1}^{|\obsSet|} \mathbb{E}\left[ \frac{1}{I_T^k}\left(\partial_\theta I_T^k \right)^2\right]\,,
    \label{eq:Fisher_information_rate_average}
\end{equation}
where the expectation value is taken with respect to all the trajectories up to time step $T$. Hence: the Fisher information rate is only a function of the individual dependence on $\theta$ of the currents; the interplay among stochastic currents of different channels does not play a role at first order. 

\section{Information loss and compression}
\label{sect:Information_postprocessing}

The measurement record $\omega_{1:N}$ in Eq.~\eqref{eq:measurement_record} represents all the information obtained from the quantum jump unraveling.
The Fisher information studied in the previous sections describes the best attainable precision, maximized over all estimator functions $\hat{\theta}(\omega_{1:N})$ that exploit the full data record. 
In practice, however, we may often not end up using the full data. 
For example, we might not have access to the jump symbols $k_i$, but only the time tags $\tau_i$, or vice versa.
Or we might want to use only a summary statistic, like the average waiting time.
Discarding part of the record in this fashion constitutes a form of information compression. 
Since we are throwing away information, the corresponding Fisher information must be smaller; i.e., the estimation precision will be reduced. 
This can be formalized by the data-processing inequality for the Fisher information~\cite{Zegers_2015, Zamir_1998, Ferrie_2014}: given any non-parameter-dependent map $\Phi$ acting on a trajectory $\omega_{1:N}$, 
\begin{equation}
    F_{\omega_{1:N}}(\theta) \geq F_{\Phi(\omega_{1:N})}(\theta)\,.
    \label{eq:data_processing_inequality}
\end{equation}
This naturally establishes a trade-off between estimation precision and memory compression.

In this section we explore the loss of FI for different choices of data compression. 
It turns out that most expressions derived in Secs.~\ref{sect:Multi_channel_renewal_processes} and~\ref{sect:the_WTD_for_non_renewal_processes} can be adapted to also encompass this possibility. 
The simplest case is that of partial monitoring; i.e., when one does not have access to all jump channels, or when they operate with finite efficiency. 
This case was already discussed around Eq.~\eqref{eq:jump_superoperators}, and simply amounts to adding an efficiency $\eta_k$ to each channel, or restricting the sum in Eq.~\eqref{eq:no_jump_op} to a subset $\obsSet$. 
One can think of the channel $\Phi$ in Eq.~\eqref{eq:data_processing_inequality} as an operation that removes entries with $k_i \notin \obsSet$ from $\omega_{1:N}$.
Partial monitoring therefore modifies the main superoperators $\mathcal{J}_k$ and $\mathcal{L}_0$. But the general structure of the trajectory distributions remain unchanged, and hence all formulas from Secs.~\ref{sect:Multi_channel_renewal_processes} and~\ref{sect:the_WTD_for_non_renewal_processes} still apply. 

Another example of data compression is when we have access only to channel symbols $k_i$, but not time-tags $\tau_i$, or vice-versa. 
Starting from the general trajectory probability~\eqref{eq:trajectory_probability} and marginalizing over all time-tags yields the symbol probability 
\begin{equation}
    \prob(k_1,\ldots,k_N) = \tr\big\{ M_{k_N} \cdots M_{k_1}\rho_0\big\},
\end{equation}
where $M_k = - \mathcal{J}_k \mathcal{L}_0^{-1}$. 
This statistics was recently studied in detail in Ref.~\cite{Landi_2023}. 
Renewal processes continues to be renewal in this case.
Hence, the corresponding FI will be given exactly by Eq.~\eqref{eq:Fisher_renewal_channels}. 
For non-renewal processes, on the other hand, one can still apply the monitoring operator formalism in Eq.~\eqref{eq:generic_MO_step}, which remains largely unchanged. All one has to do is replace  $\mathcal{M}_x$ with $M_k$.

If we instead marginalize~\eqref{eq:trajectory_probability} over the jump channels we obtain 
\begin{equation}\label{eq:trajectory_probability_no_channel_symbols}
    \prob(\tau_1,\ldots,\tau_N) = \tr\big\{ \mathcal{J} e^{\mathcal{L}_0 \tau} \cdots \mathcal{J} e^{\mathcal{L}_0 \tau}\rho_0\big\},
\end{equation}
where $\mathcal{J} = \sum_{k\in \obsSet} \mathcal{J}_k$. 
Notice how this result still depends on the choice of monitored channels $\obsSet$, which influences both $\mathcal{J}$ and $\mathcal{L}_0$.
Eq.~\eqref{eq:trajectory_probability_no_channel_symbols} therefore assumes one does not know which specific channel clicked, but one knows it belongs to $\obsSet$. 
In general, Eq.~\eqref{eq:trajectory_probability_no_channel_symbols} will no longer be renewal, even if the original process was.
The reason is because the renewal condition~\eqref{eq:renewal_condition} is not, in general, satisfied for $\mathcal{J}$. 
Nonetheless, the monitoring operator formalism in Eq.~\eqref{eq:generic_MO_step} remains perfectly valid, provided one replaces $\mathcal{M}_x \to \mathcal{J} e^{\mathcal{L}_0 \tau}$. 

So far, we have assumed compression schemes in which one has limited access to one aspect of the dataset or another, and which are specific to the form of the jump trajectory outcomes. 
Another important compression is to consider instead only summary statistics, analyzed in detail in Ref.~\cite{Smiga_2023}. In order to draw a relation, we here add an explanation of the effect of using such summary statistics for data originating from quantum jump trajectories. The broader case, independent of quantum jump dynamics, is discussed in Ref.~\cite{Smiga_2023}.

Suppose the process has a single jump channel. 
Instead of considering estimators $\hat{\theta}(\tau_1,\ldots,\tau_N)$ that are arbitrary functions of the entire dataset, one might consider estimators $\hat{\theta}(T)$ that depend only on the sample mean 
\begin{equation}
    T = \frac{\tau_1+\cdots+\tau_N}{N}. 
\end{equation}
This is another form of data compression and therefore must necessarily lead to a smaller Fisher information. 
To calculate this FI one must take into account the fact that the waiting times can be correlated with each other. 
As shown in~\cite{Radaelli_2023,Smiga_2023} the asymptotic FI, for large $N$, will be 
\begin{equation}\label{eq:FI_sample_mean}
    F_T(\theta) \simeq N \left(\frac{\partial\mu}{\partial\theta}\right)^2 \frac{1}{\sigma^2+ 2 \sum_{i=1}^{N-1} C_i},
\end{equation}
where $\mu = \mathbb{E}[\tau_i]$ and $\sigma^2 = \var[\tau_i]$ are the mean and variance of the individual waiting times (assumed stationary), while $C_{i-j} = {\rm cov}(\tau_i,\tau_j)$ is the covariance between two waiting times. Note that~$C$ is represented with a single index because stationarity (time-translation invariance) is assumed. For renewal processes the waiting times become i.i.d.~and $C_i = 0$. 
In this case, the WTD will be given by Eq.~\eqref{eq:wtd_renewal}: $W(\tau) = \tr\{ \mathcal{J} e^{\mathcal{L}_0\tau} \sigma\}$. 
And the mean and variance in Eq.~\eqref{eq:FI_sample_mean} can be computed from the Laplace transform 
\begin{equation}
    \tilde{W}(s) = \int_0^\infty d\tau~W(\tau) e^{-s \tau} = \tr\big\{\mathcal{J} (s-\mathcal{L}_0)^{-1} \sigma\big\},
\end{equation}
as $\mu = - \tilde{W}'(0)$ and $\sigma^2 = \tilde{W}''(0) - \mu^2$. 
Conversely, for non-renewal processes $C_i \neq 0$ and the correlations will play a non-trivial role, since the second term in the denominator of Eq.~\eqref{eq:FI_sample_mean} can be positive or negative, depending on the model. 
This is explored in more details in Refs.~\cite{Radaelli_2023,Smiga_2023}.

\section{Examples}\label{sec:examples}
In this section, we consider several examples to illustrate concepts explored throughout the paper. The examples are summarized in Table~\ref{tab:ex_summary}.

\begin{table}[htbp]

    \centering
    \begin{tabular}{llcc}
        \hline\hline 
        Example & & Renewal & Multi-channel\\
        \hline
        Q. therm. & \ref{sec:qubit_thermometry} & X & X \\
        Res. fluor. & \ref{sec:resonant_fluorescence} & X &  \\
        Coupled q. & \ref{sect:coupled_qubits} &  &  \\
        Micromaser & \ref{sec:maser} &  & X \\
        \hline\hline
    \end{tabular}
        \caption{Summary of presented examples.}
    \label{tab:ex_summary}
\end{table}
\vspace{-2pt}

\subsection{Qubit thermometry} \label{sec:qubit_thermometry}

We consider a qubit coupled to a thermal bath. The Hamiltonian of the isolated qubit system is
\begin{equation}\label{eq:qubit_H}
    H = \frac{\omega}{2} \sigma_z + \frac{\Omega}{2}\sigma_x\,,
\end{equation}
where $\omega$ is the detuning and $\Omega$ is the Rabi frequency. The coupling with the thermal bath is described by two jump operators, of the form
\begin{subequations}\label{eq:qubit_Ls}
    \begin{align}
        L_+ &= \sqrt{\gamma \bar{n}} \sigma_+ \\
        L_- &= \sqrt{\gamma (\bar{n}+1)} \sigma_-\,,
    \end{align}
\end{subequations}
where $\gamma$ is the emission rate and $\bar{n}$ is the average occupation number of the relevant mode in the environment; the value of $\bar{n}$ is connected to the temperature of the environment via the Bose-Einstein distribution. We assume that both the jumps $L_+$ and $L_-$ are observable, and we would like to estimate the parameter $\bar{n}$, of interest in thermometry applications. 
This process is renewal. \\

\vspace{-12pt}

The Fisher information about $\bar{n}$ contained in the measurement record can be computed as in Eq.~\eqref{eq:Fisher_renewal}, and yields
\begin{widetext}\begin{align}
    F_{1:N}(\bar{n}) =& \frac{N}{2\bar{n}(\bar{n}+1)} \left[ \frac{(\bar{n}+1)^2 + \bar{n}^2}{\bar{n}(\bar{n}+1) + \frac{(\Omega/\gamma)^2}{2} \frac{1}{1+4\left( \frac{\omega/\gamma}{2\bar{n}+1} \right)^2}}+\frac{(\Omega/\gamma)^2}{\bar{n}(\bar{n}+1)\left( 1+4\left( \frac{\omega/\gamma}{2\bar{n}+1} \right)^2 \right) + \frac{(\Omega/\gamma)^2}{2}} \right]\,.
\end{align}\end{widetext}
WTD and Fisher information are shown in Fig.~\ref{fig:FI_qubit_thermometry}. 

\begin{figure}[hbt]
\centering
\includegraphics[width=\linewidth]{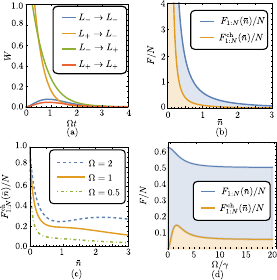}
    \caption{Thermometry with a qubit with observed raising and lowering operations. (a)~The WTD for different observations. Parameters: $\omega=\Omega=\gamma=1$ and $\bar{n}=1.5$. (b)~The Fisher information per observed jumps. Parameters: $\omega=\Omega=\gamma=1$. The total rate (blue) is shown with the contribution from the channels $F^{\rm ch}$ alone (orange). (c)~Proportion of total Fisher information in the channels, alone, for different Rabi oscillations $\Omega$ and $\omega=\gamma=1$. (d)~The Fisher information per observed jumps as a function of~$\Omega$, same parameters as in (b).}
    \label{fig:FI_qubit_thermometry}
\end{figure}

A special remark is in order for the limit $\Omega\to 0$ (see Fig.~\ref{fig:FI_qubit_thermometry}d). In this case, the jump operators $L_+$ and $L_-$ satisfy the conditions for an exponential WTD discussed in Sec.~\ref{sect:particular_cases} (the final states are orthogonal among each other, and the Hamiltonian commutes with the density matrix of the final states), as specified in the previous section. The only non-zero WTDs are
\begin{subequations}
\begin{align}
    W(\tau, L_+|L_-) &= \gamma \bar{n} e^{-\gamma \bar{n} \tau}\,; \\
    W(\tau, L_-|L_+) &= \gamma (\bar{n} + 1) e^{-\gamma (\bar{n}+1)\tau}\,.
\end{align}
\end{subequations}
The Fisher information considerably simplifies, yielding
\begin{equation}
    F_{1:N}(\bar{n},\Omega=0) = \frac{N}{2}\left[ \frac{1}{\bar{n}} + \frac{1}{\bar{n}+1} \right]\,.
\end{equation}
If $\Omega \to 0$, the jump channels have to alternate between $L_+$ and $L_-$ (no two jumps of the same kind can happen in a row). Therefore, no information about $\bar{n}$ can be contained in the sequence of the jump channels, when one ignores the jump times: $F_{1:N}(\bar{n}) = F^{\rm times|ch}_{1:N}(\bar{n})$.

The relative contribution of $F^{\rm ch}_{1:N}$ to the total Fisher information is illustrated in Fig.~\ref{fig:FI_qubit_thermometry}b--c. One finds that, with all else equal, the sequence of observed channels dominates the Fisher information of the temperature for a cold system ($\bar{n}\to 0$) but is generally less important as temperature increases. Further, one observes that $F^{\rm ch}_{1:N}\to 0$ as $\Omega\to0$ (for finite temperature, $\bar{n}>0$), reflecting the fact that $F_{1:N}(\bar{n}) = F^{\rm times|ch}_{1:N}(\bar{n})$ in this case. 

\subsection{Resonant fluorescence} \label{sec:resonant_fluorescence}

Consider again Eqs.~\eqref{eq:qubit_H} and~\eqref{eq:qubit_Ls}, but assume $\omega = \bar{n} = 0$, so only emissions, $L = \sqrt{\Gamma} \sigma_-$, are observed. 

The expression for the Fisher information w.r.t. the Rabi frequency $\Omega$ has been shown to be~\cite{Kiilerich_2014}
\begin{equation}
    F(\Omega) = N \left(\frac{8}{\Gamma^2} + \frac{4}{\Omega^2}\right)\,,
\end{equation}
where $N$ is the number of recorded jumps. The computation of the Fisher information of the sample mean via Eq.~\eqref{eq:FI_sample_mean} yields
\begin{equation}
    F_{T}(\Omega) = \frac{N}{\Omega^2}\frac{4}{1 - 2 (\Omega/\Gamma)^2+4 (\Omega/\Gamma)^4}\,.
\end{equation}
In Fig.~\ref{fig:sample_mean}, we represent the Fisher information for the entire measurement record, and for the sample mean for resonant fluorescence. We notice how, in general, the Fisher information of the sample mean is not comparable to the one of the full measurement record; while the former lower bounds the latter, the bound is clearly far from tight.

\begin{figure}[htbp]
    \centering
    \includegraphics[width=\linewidth]{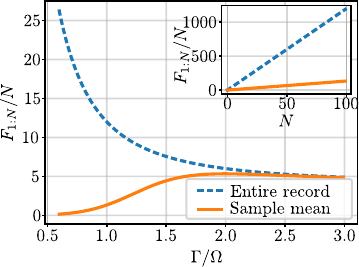}
    \caption{Fisher information of the entire measurement record, and of the sample mean for the resonant fluorescence process. In the main panel, the Fisher information rate for the Rabi frequency $\Omega = 1$ is represented as a function of the jump rate $\Gamma$; in the inset plot, for fixed $\Gamma = \Omega = 1$, as a function of $N$, the number of recorded jumps.}
    \label{fig:sample_mean}
\end{figure}

\subsection{Coupled qubits} \label{sect:coupled_qubits}
\begin{figure}[htbp]
    \centering
    \includegraphics[width=\linewidth]{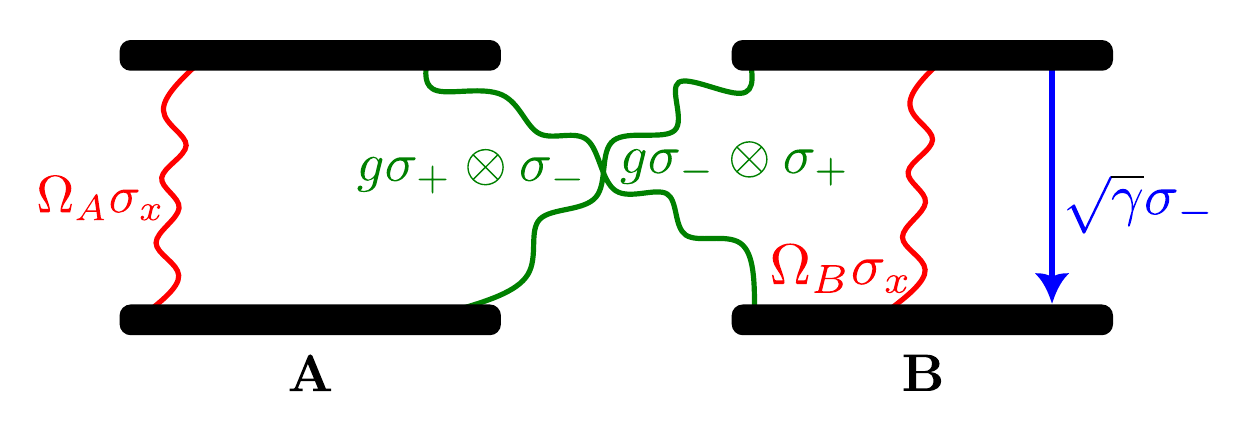}
    \caption{A depiction of the coupled-qubit model.}
    \label{fig:coupled_qubit_schematic}
\end{figure}
As a paradigmatic case of a non-renewal process, we consider two qubits, both undergoing Rabi oscillations, interacting with an exchange Hamiltonian, and such that one of them can ``leak'' excitations into a thermal bath. The Hamiltonian of the system has the form
\begin{equation}
\begin{split}
    H  = & \Omega_A\sigma_x^A + \Omega_B\sigma_x^B + \omega_A \sigma_z^A + \omega_B\sigma_z^B \\& + g \left(\sigma_+^A \sigma_-^B + \sigma_-^A \sigma_+^B\right)\,,
\end{split}
\end{equation}
where $\Omega_{A,B}$ are Rabi frequencies, $\omega_{A,B}$ are detuning terms, and $g$ parametrizes the coupling strength. A tensor product with the identity is implied where not otherwise specified. There is only one jump operator, given by the emission from qubit $B$ into the bath; the corresponding jump operator is
\begin{equation}
    L = \sqrt{\gamma} \sigma_-^B\,.
\end{equation} 

The model is depicted in Fig.~\ref{fig:coupled_qubit_schematic}, and an example of a trajectory is given in \ref{fig:coupled_qubits_model}a. It is clear that the model is non-renewal: while at each jump the state of qubit $B$ is completely reset to the ground state, the state of qubit $A$ is affected by the jump (which breaks all the entanglement between the qubits created by the Hamiltonian evolution), but not reset: there is quantum memory carried on at each jump.

We used the Fisher-Gillespie algorithm to compute the Fisher information of the measurement record for the estimation of $\gamma$ in the coupled qubit model. The result is plotted in Fig.~\ref{fig:coupled_qubits_model}. It is apparent that the Fisher information scales asymptotically linearly with time; this is expected for any stochastic process whose correlations decay with time~\cite{Radaelli_2023}.

In Fig.~\ref{fig:coupled_qubits_model}c, we also highlighted one of the stochastic trajectories of $\tr[\xi_t]^2$. Remarkably, while the Fisher information (given by $\mathbb{E}[\tr[\xi_t]^2]$) always monotonically increases with time (up to small statistical fluctuations), on individual trajectories a decrease can be observed. In particular, the squared trace of the monitoring operator can be interpreted as an indicator of how atypical a trajectory is. Indeed, on a typical trajectory one would expect MLE to retrieve a very close estimate~$\hat{\theta}$ to~$\theta$. By the MLE condition,~$\text{Tr}[\xi_{t}]^{2}=0$. Hence, we expect $\text{Tr}[\xi_{t}]^{2}$ to be small for typical trajectories, and much larger for atypical trajectories.

\vspace{48pt}

\begin{figure*}[tp]
    \centering
    \includegraphics[width=0.75\textwidth]{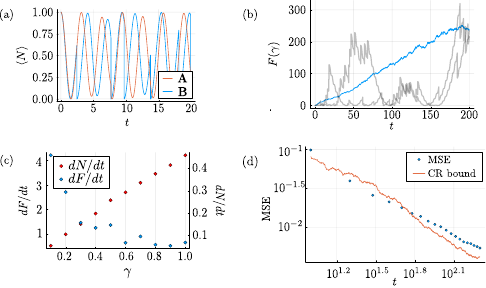} 

    \caption{(a): a stochastic quantum trajectory for the two-qubit model, in terms of the expectation value of the number operator for each qubit. Note that, whereas qubit $B$ undergoes a complete reset after each jump, qubit $A$, while affected by the jump, is not reset: this behavior clearly expresses the non-renewal nature of the process. To highlight this effect, the plot is obtained for an increased value~$g = 0.1$, with respect to the parameters below. (b): the result of a Gillespie-Fisher algorithm calculation of the Fisher information for $\gamma$ along the jump process. The average on the 2000 trajectories is represented by the red line. Some of the trajectories for $\tr[\xi_t]^2$ are represented in the background. (c): the asymptotic Fisher information rate for $\gamma$, as a function of~$\gamma$ itself, calculated via a linear fit of the Fisher information, represented along side the average jump rate. (d): mean squared error of MLE estimation for~$\gamma$, compared with the numerical Cramér-Rao bound for 2000 trajectories. (Parameters: $\gamma = 0.4$, $\bar{n}_{th} = 1$, $\Omega_{A} = 1$, $\Omega_{B} = 1$, $g = 0.01$)}
    \label{fig:coupled_qubits_model}
\end{figure*}

\subsection{Single-atom maser}
\label{sec:maser}

\begin{figure}[htbp]
    \centering
    \includegraphics[width=0.8\linewidth]{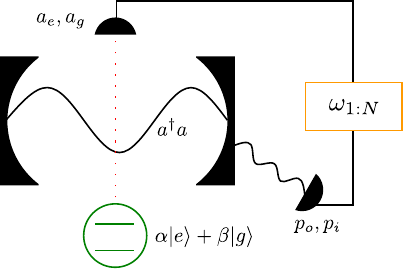}
    \caption{Illustration of the single-atom maser model. A beam of two-level atoms passes through a cavity where each atom interacts with the electromagnetic field. The energy of the atoms is measured when it exits the cavity. The cavity itself is coupled with the environment modeled as a thermal bath.}
    \label{fig:schematic_micromaser}
\end{figure}

In this example, we apply the previously discussed concepts to the study of a model of relevant physical significance, the single-atom maser or micromaser~\cite{Scully_1997}, depicted in Fig.~\ref{fig:schematic_micromaser}. While being developed in the last decades of the past century, the model is still the object of intensive study~\cite{Korkmaz_2023, Feyisa_2023, Mikhalychev_2022}, and remains paradigmatic in the context of collision models. Physically, it consists of an optical cavity of high quality factor, across which a beam of highly excited Rydberg atoms passes. The typical energy of the excited states of those atoms with respect to their ground state is of the order of tens of GHz, i.e., in the microwave regime.

\begin{figure*}[bt]
    \centering
    \includegraphics[width=\textwidth]{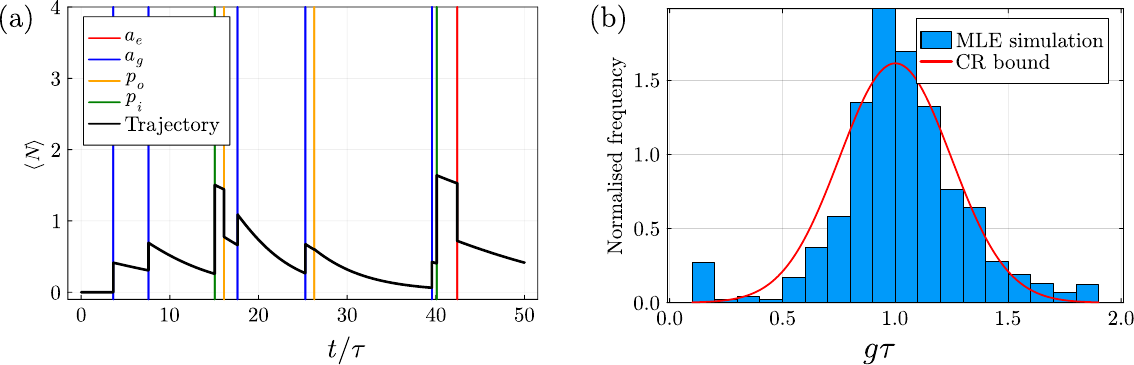}
    \caption{(a) Example of a trajectory for the number operator on the cavity, the different jump channels are highlighted. (b) Results of an MLE estimation task for $g$, compared with the Cramér-Rao bound obtained via the Fisher-Gillespie algorithm. Note that, formally, the convergence to a normal probability distribution is only guaranteed for an asymptotically large number of data points. Here, the convergence can already be observed for the final time considered for this plot. Parameters: $g = 1$, $\tau=1$, $\theta = \pi/4$, $\bar{n}_{\rm th} = 0.1$, $\gamma = 0.1$, e.m. field capped at the first 5 levels. Simulation obtained with 1000 trajectories, capped at final time $100\tau$.}
    \label{fig:drawing_micromaser}
\end{figure*}

A beam of two-level atoms, all in the same initial state $\ket{\psi}= \alpha\ket{e} + \beta \ket{g}$, fly through a cavity. The beam is sparse enough that, on average, there is at most one atom inside the cavity. The atom interacts with the electromagnetic field inside the cavity via the Jaynes-Cummings Hamiltonian
\begin{equation}
    H = \frac{g}{2}\left[a  \sigma_+ + a^\dagger  \sigma_-\right]\,,
    \label{eq:Jaynes_Cummings}
\end{equation}
where $a$ is the annihilation operator for the single-mode field. As above, $\sigma_+ = \ket{e}\bra{g}$, $\sigma_- = \ket{g}\bra{e}$; $g$ is the coupling constant. This Hamiltonian acts while the atom is inside the cavity, for time $\tau$. Upon exiting the cavity, the atom is measured in the energy basis (projective measurement with $\ket{e}\bra{e}$ and $\ket{g}\bra{g}$). 

The cavity itself is coupled to the environment, represented by a thermal bath; hence, photons leak from the cavity to the environment and enter the cavity from the environment.

In the analytical treatment of the micromaser model, one usually considers only the degrees of freedom of the field in the cavity, while the atomic ones are traced out. In this case, there are four possible incoherent jump channels, which we label as:
\begin{itemize}
    \item $a_e$ and $a_g$, corresponding to the atom exiting the cavity being detected respectively in the excited and the ground state;
    \item $p_i$ and $p_o$, corresponding to a photon entering the cavity from the thermal bath or being emitted from the cavity to the thermal bath, respectively. 
\end{itemize}
We can obtain the jump operators corresponding to the atomic measurements by tracing out the atomic d.o.f. from the unitary corresponding to the Jaynes-Cummings Hamiltonian~\cite{Scully_1997}. Denoting $\mathbf{\tilde s}=\sqrt{a^\dag a}$ and $\mathbf{s}=\sqrt{a^\dag a+1}$ this yields:
\begin{subequations}
\begin{align}
    L_{a_e} &= \alpha \cos\left(g\tau \mathbf{s}\right) -  i \beta \sin\left(g\tau\mathbf{s}\right)\mathbf{s}^{-1} a\,,\\
    L_{a_g} &= \alpha a^\dagger \sin\left(g\tau\mathbf{s}\right)\mathbf{s}^{-1} + i \beta \cos\left(g\tau\mathbf{\tilde s}\right)\,.
\end{align}
\end{subequations}

The jump operators corresponding to the absorption by and emission from the cavity are given by
\begin{subequations}
\begin{align}
    L_{p_i} &= \sqrt{\gamma \bar{n}_{\rm th}} a^\dagger\,, \\
    L_{p_o} &= \sqrt{\gamma (\bar{n}_{\rm th} + 1)} a\,,
\end{align}
\end{subequations}
where $\bar{n}_{\rm th}$ is the average number of photons given by the Bose-Einstein distribution for the specific temperature of the bath, and $\gamma$ is the emission rate. 

This model is clearly non-renewal, as no jumps are able to reset the state to a fixed post-jump state. An example of a quantum trajectory for the model is shown in the right panel of Fig.~\ref{fig:drawing_micromaser}.

The increased calculation efficiency given by the Gillespie-Fisher algorithm makes studying the Fisher information of a physical system such as the single-atom maser much more accessible. In particular, we are interested in the effect of the initial quantum state of the flying atoms on the Fisher information rate for the estimation of the coupling constant $g$. Let us define $\theta =\arccos(|\alpha|)$, so that $\ket{\psi} = \cos\theta \ket{e} + \sin\theta \ket{g}$, where we set the relative phase to zero. The behavior of the Fisher information rate for $g$, as a function of $\theta$, is shown in Fig.~\ref{fig:relevance_theta}.

\begin{figure}[htbp]
    \centering
    \includegraphics[width=\columnwidth]{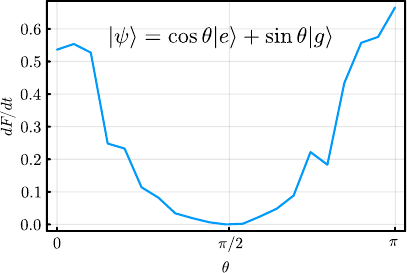}
    \caption{The Fisher information rate, as a function of $\theta$, for the single-atom maser system. For each $\theta$, the simulations were performed with $g=1$, $\tau = 1$,
     $\gamma=0.1$, $\bar{n}_{\rm th}=0.1$, the e.m. field limited to the first~5 levels, 100~trajectories, a time interval of $dt=0.001$, and $g$ variation for the derivatives $dg=0.01$.
     }
    \label{fig:relevance_theta}
\end{figure}

We observe that the highest Fisher information rate is obtained for $\theta\sim 0$ and $\theta \sim \pi$, corresponding to the incoming atoms being almost completely in the energy eigenstate $\ket{e}$. In contrast, the Fisher information rate appears to vanish for $\theta \sim \pi/2$, corresponding to the incoming atoms being almost completely in state $\ket{g}$, meaning that, in this regime, negligible information about $g$ is contained in the measurement record. The state of the field is such that there is on average much less than one photon inside the cavity:
\begin{equation}
    \ket{f} \sim (1-\varepsilon) \ket{0} + \varepsilon \ket{1}\,,
\end{equation}
with $\varepsilon \ll 1$. Applying the unitary evolution given by the Jaynes-Cummings Hamiltonian in Eq.~\eqref{eq:Jaynes_Cummings} at first order, $U = \mathbb{I} - i H \tau$, one obtains a post-interaction state
\begin{equation}
    \ket{\Psi} = \left(\mathbb{I} - iH\tau\right)\ket{f} \otimes \ket{\psi}\,.
\end{equation}
If $\ket{\psi}=\ket{g}$, then $\ket{\Psi} = (1-\varepsilon)\ket{0}\otimes\ket{g} + \varepsilon g/2 \ket{0}\otimes\ket{e}$, and we observe that the term bearing the $g$-dependence (hence the Fisher information) is suppressed in the $\varepsilon \to 0$ limit. On the contrary, if $\ket{\psi} = \ket{e}$, then $\ket{\Psi} \sim \ket{0} \otimes \ket{e} - i g\tau/2 \ket{1}\otimes\ket{g}$ in the $\varepsilon\to 0$ limit: the $g$-dependence is not suppressed. In other words, since there are almost always zero photons in the cavity field, some meaningful interaction can be obtained only if there is a significant overlap between the state of the incoming atom and $\ket{e}$, otherwise no interaction takes place, and the atom flies out of the cavity in the very same state in which it flew in. Since the $g$-dependence is only picked up if the interaction happens, atoms starting in state $\ket{g}$ will not allow for any information about $g$ to be collected.

\section{Conclusions}
In this work, we provided a complete framework for the estimation of parameters encoded in quantum jump processes, applicable in scenarios where an experimenter has access to jump times and multiple jump channels.
First, we showed the asymptotic behavior of the Fisher information for multi-channel renewal processes in terms of the waiting time distribution, generalizing  the results presented in Ref.~\cite{Kiilerich_2014}, and discussing the partition of information between jump channels and jump times.
We then discussed the more intricate case of non-renewal processes. Our treatment is based on the monitoring operator formalism, first introduced in Ref.~\cite{Gammelmark_2013}, and computationally enhanced in Ref.~\cite{Albarelli_2018}. Here, we introduced a new evolution method, the Gillespie-Fisher algorithm based on the quantum Gillespie method~\cite{Radaelli_2023b}, which allows for an efficient calculation of the Fisher information on quantum trajectories in many cases of interest. As we showed, the Fisher information rate can be decomposed in terms of the stochastic currents, elucidating how it evolves in time. We showed how the monitoring operator can be employed to perform maximum likelihood estimation on quantum jump trajectories in a numerically stable way. 
To understand the behavior of the Fisher information under classical post-processing of the measurement record we provided an analysis based on the the data-processing inequality.
Finally, we presented a number of physical examples, of different levels of physical complexity, to showcase our findings.

These results provide a comprehensive set of tools for describing metrology of quantum jumps in any platform. 

In recent years, metrological techniques based on the Fisher information have been extensively exploited in quantum thermodynamics, for example in the context of thermodynamic and kinetic uncertainty relations~\cite{Hasegawa_2019,VanVu_2022,Moreira_2025}. Our results, therefore, constitute the ideal toolbox for the computation of these inequalities in quantum jump processes.

We expect these tools to have a meaningful impact for developing new experiments --- such as developing novel quantum sensors --- as well as theoretical work for understanding optimal methods for encoding information.

\section*{Acknowledgments}
The authors thank Gerard Milburn for suggesting the single atom maser example, and Yoshihiko Hasegawa for insightful discussions about the waiting time distribution. MR thanks the ToCQS group at Trinity College Dublin for access to their computational resources, and the Irish Centre for High-End Computing (ICHEC) for the provision of computational facilities and support. The research conducted in this publication was funded by the Irish Research Council under grant numbers IRCLA/2022/3922 and GOIPG/2022/2321. This publication was made possible through the support of Grant 62423 from the John Templeton Foundation. The opinions expressed in this publication are those of the author(s) and do not necessarily reflect the views of the John Templeton Foundation.

\bibliographystyle{quantum}
\bibliography{bibliography}

\clearpage
\appendix

\begin{widetext}

\section{Form of the jump operators for renewal processes}
\label{sect:form_of_the_jump_operators_for_renewal_processes}
In this appendix, we prove that, in a renewal master equation in which all the jumps are monitored, the jump operators have to have a specific form. In particular, let $L_k$ be a jump operator satisfying the renewal condition
\begin{equation}
    \frac{L_k\rho L_k^\dagger}{\tr\left[L_k \rho L_k^\dagger\right]}=\sigma_k\,,
\end{equation}
where $\sigma_k$ does not depend on $\rho$. Then, we prove that $L_k = c_k \ket{\mu_k}\bra{\nu_k}$, where $c_k=\sqrt{\gamma_k}$.

Let us perform a singular value decomposition of the operator $L_k$:
\begin{equation}
    L_k = \sum_\alpha c_{k,\alpha} \ket{\mu_{k,\alpha}}\bra{\nu_{k,\alpha}}\,,
\end{equation}
where $\{c_{k,\alpha}\}$ are the singular values of $L_k$, and $\{\ket{\mu_{k,\alpha}}\}$ and $\{\ket{\nu_{k,\alpha}}\}$ are, respectively, the orthonormal bases of the left-singular and right-singular vectors. 

The ratio defining the renewal condition can be written as
\begin{equation}
    \frac{L_k\rho L_k^\dagger}{\tr[L_k\rho L_k^\dagger]} = \sum_{\alpha \beta} \frac{c_{k,\alpha}c_{k,\beta}^* \ket{\mu_{k,\alpha}}\bra{\nu_{k,\alpha}}\rho \ket{\nu_{k,\beta}}\bra{\mu_{k,\beta}}}{\sum_a |c_{k,a}|^2 \bra{\nu_{k,\alpha}}\rho\ket{\nu_{k,\alpha}}}  = \sigma_k\,.
\end{equation}
We observe that, in order for $\sigma_k$ not to depend on $\rho$, only one of the singular values $\{c_{k,\alpha}\}$ can be non-zero. It follows that the jump operator $L_k$ can be written in the simpler form
\begin{equation}
    L_k = c_k \ket{\mu_k}\bra{\nu_k}\,.
\end{equation}
By plugging this equation again in the definition of renewal jump, we obtain that $\sigma_k = \ket{\mu_k}\bra{\mu_k}$. This implies that, in a renewal process in which all jumps are monitored, the destination state of every jump has to be a pure state. 

\section{Fisher information for renewal processes}
\label{sect:Fisher_information_for_renewal_processes}
In this Appendix, we prove the expression for the Fisher information of renewal processes, Eq.~\eqref{eq:Fisher_renewal}. We first show an intermediate result on the dynamical activity of channels.

\begin{lemma}
    Let $\rho$ be the steady state of a GKSL master equation, such that $\mathcal{L}\rho =0$. Let $f_k = \tr[\mathcal{J}_k\rho]$ be the dynamical activity of channel $k$ and define 
    \begin{equation}\label{eq:p_k_appendix_markov}
        p_k = \frac{f_k}{\sum_j f_j}\,.
    \end{equation}
    For renewal processes, $p_k$ (or $f_k$) is the steady state of the Markov chain 
    \begin{equation}
        \sum_{i \in \obsSet} p(j|i)f_i = f_j\,,
        \label{eq:steady_state_Markov_chain}
    \end{equation}
    where
    \begin{equation}
        p(j|i) = \int_0^\infty d\tau \tr\left[\mathcal{J}_j e^{\mathcal{L}_0\tau}\sigma_i\right] = \int_0^\infty d\tau W(\tau,j|i)
    \end{equation}
    is the probability of having a jump in channel $j$ after a jump in channel $i$, with no other jumps in between, regardless of how much time elapses. 
\end{lemma}
\begin{proof}
    It suffices to prove the result for $f_k$. 
    We assume that $\mathcal{L}_0$ is invertible; this condition is tantamount to the normalisability of the waiting time distribution, and means that a jump will always occur in the end with unit probability. Then, $p(j|i)$ can be rewritten as
    \begin{equation}
        p(j|i) = - \tr\left[\mathcal{J}_j \mathcal{L}_0^{-1} \sigma_i\right]\,.
    \end{equation}
    By Eq.~\eqref{eq:renewal_condition} one then has that 
    \begin{align*}
        \sum_{i\in \obsSet} p(j|i) f_i &= - \sum_{i \in \obsSet} \tr\big\{ \mathcal{J}_j \mathcal{L}_0^{-1} \sigma_i\big\} \tr\{ \mathcal{J}_i \rho_{\rm ss}\} 
        \\
        &= - \sum_{i \in \obsSet} \tr\big\{ \mathcal{J}_j \mathcal{L}_0^{-1} \mathcal{J}_i \rho_{\rm ss}\big\}\,.
    \end{align*}
    We now use the fact that $\sum_{i \in \obsSet} \mathcal{J}_i = \mathcal{L}-\mathcal{L}_0$, and $\mathcal{L}\rho_{\rm ss} = 0$. Hence, we arrive at 
    $\sum_{i\in \obsSet} p(j|i) f_i = \tr\{\mathcal{J}_j \rho_{\rm ss}\} = f_j$.
\end{proof}

As defined, the measurement record probability distribution $\prob(\omega_{1:N})$ in Eq.~\eqref{eq:probability_measurement_record_renewal} is not a stationary process. 
This is because it is conditioned on an arbitrary initial channel $k_0$. 
We can make it stationary if we modify it as 
\begin{equation}
    \prob(\omega_{1:N}) = \sum_{k_0} W(\tau_N, k_N | k_{n-1}) \cdots W(\tau_1, k_1 | k_0) p_{k_0}\,,
    \label{eq:renewal_P_new_dist_stationary}
\end{equation}
where $p_{k_0}$ is given in Eq.~\eqref{eq:p_k_appendix_markov}.
To illustrate the stationarity, consider $\omega_2 = (k_1,\tau_1,k_2,\tau_2)$. Then 
\begin{align*}
    \sum_{k_1} \int\limits_0^\infty d\tau_1 P(k_1,\tau_1,k_2,\tau_2) &= \sum_{k_1} \int\limits_0^\infty d\tau_1  \sum_{k_0} W(\tau_2,k_2|k_1) W(\tau_1,k_1|k_0) p_{k_0}
    \\
    &= \sum_{k_0,k_1} W(\tau_2,k_2|k_1) p(k_1|k_0) p_{k_0}
    \\
    &= \sum_{k_1} W(\tau_2,k_2|k_1) p_{k_1} = P(k_2,\tau_2)\,.
\end{align*}
An important consequence of stationarity is the following lemma.
\begin{lemma}
    Let $g(\tau_i, k_{i-1}, k_i)$ denote a function depending only on two consecutive jumps, and the time $\tau_i$ between them. Then, any average over a $N$-steps trajectory simplifies to   
    \begin{equation}
    \begin{split}
        \mathbb{E}\left[g(\tau_i, k_{i-1}, k_i)\right] & = \sum_{k_0\ldots k_N \in \obsSet} \int_0^\infty d\tau_1 \cdots \int_0^\infty d\tau_N \prob\left(k_0, (\tau_1, k_1),\ldots, (\tau_N, k_N) \right) g(\tau_i, k_{i-1}, k_i) \\
        & = \sum_{k_{i-1}, k_i \in \obsSet} \int_0^\infty d\tau_i W(\tau_i, k_i|k_{i-1}) p_{k_{i-1}} g(\tau_i, k_{i-1}, k_i)\,.
    \end{split}
    \label{eq:stationarity_property_g}
    \end{equation}
\end{lemma}

We can now turn to the calculation of the Fisher information in a measurement record $\omega_{1:N}$ for a renewal process. Combining Eq.~\eqref{eq:Fisher_information} and Eq.~\eqref{eq:renewal_P_new_dist_stationary} we get
\begin{equation}
    F_{1:N}(\theta) = - \mathbb{E}\left[\sum_{j=2}^N \frac{\partial^2}{\partial\theta^2} \log W(\tau_j,k_j|k_{j-1}) + \frac{\partial^2}{\partial\theta^2}\log \sum_{k_0} W(\tau_1,k_1|k_{0}) p_{k_0}\right]\,.
\end{equation}
The last term stems from our modification of $\prob(\omega_{1:N})$ in Eq.~\eqref{eq:renewal_P_new_dist_stationary}. Notice that this is independent of $N$, while the first term will be extensive in $N$. 
The second term, therefore, leads to a negligible contribution. 
As for the first term, we can  use Eq.~\eqref{eq:stationarity_property_g}. 
This will give rise to $N$ identical terms, leading us exactly to  the expression in Eq.~\eqref{eq:Fisher_renewal}.

\section{MLE estimation for the qubit thermometry process}
\label{app:MLE_qubit_thermo}
In this Appendix, we propose an estimation task for a renewal system, to illustrate the results of Sec.~\ref{sect:Multi_channel_renewal_processes}. We consider the qubit thermometer, whose Hamiltonian and jump operators are given in Sec.~\ref{sec:qubit_thermometry}, and we aim to estimating the average occupation number $\bar{n}$. 

\begin{figure}[htbp]
    \centering
    \includegraphics[width=\textwidth]{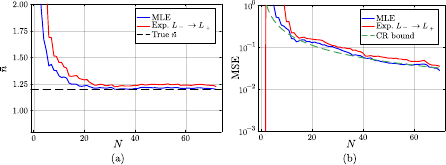}
    \caption{An estimation task performed on the qubit thermometer. On the left, the expectation value of the estimated $\bar{n}$ as a function of the number of considered jumps; on the right, the mean squared error for the two estimation strategies. Parameters: $\omega=1$, $\Omega=0.1$, $\gamma=0.02$, $\bar{n}=1.2$; 100 trajectories.}
    \label{fig:MLE_renewal}
\end{figure}

For a renewal system, the log-likelihood function is easily  expressed in terms of the WTD, as in Eq.~\eqref{eq:loglikelihood_renewal}; as expected, MLE estimation asymptotically saturates the Cram\'er-Rao bound (see Fig.~\ref{fig:MLE_renewal}). 

We also consider a different estimator, given by the average waiting time between a $L_-$ and a $L_+$ jumps. The expectation value of the waiting time between these two jumps is given by
\begin{equation}
    \langle \tau \rangle_{-+} = \int_{0}^\infty d\tau \tau W(\tau|-,+) = - \frac{\tr\left[\jump{+}\left(\mathcal{L}_0\right)^{-2}\jump{-}\rho_0\right]}{\tr\left[\jump{+}\left(\mathcal{L}_0\right)^{-1}\mathcal{J}_-\rho_0\right]},
\end{equation}
and is, of course, a function of $\bar{n}$. Given a measurement record, one can compute the average waiting time between two jumps $L_-$, $L_+$, then numerically invert $\langle \tau\rangle_{-+}$ as a function of $\bar{n}$ to find the estimate. Clearly, the performances of the estimator based on the average waiting time are worse than the ones of the MLE, because the former does not have access to all the information contained in the shape of the waiting time distributions.

\section{Single exponential WTD and classical master equations}
\label{sect:particular_cases}
In general, carrying out the integral in Eq.~\eqref{eq:Fisher_renewal} is difficult because the WTD is usually a sum of (possibly complex) exponentials in $\tau$. 
A case that offers a significant simplification is when it is a single decaying exponential:
\begin{equation}\label{eq:renewal_single_exponential}
    W(\tau,k|q) = p(k|q) b_{kq} e^{-b_{kq} \tau}\,,
\end{equation}
for coefficients $p(k|q)>0$ and $b_{kq}>0$ that may both depend on $\theta$. 
This, as we will discuss below, is the case of classical (Pauli) master equations describing incoherent dynamics. 
The coefficient $p(k|q)$ is defined to match 
Eq.~\eqref{eq:renewal_markov_transition}, so that normalization implies $\sum_k p(k|q) = 1$ for all $q$. 
Carrying out the integral in Eq.~\eqref{eq:Fisher_renewal} allows us to write 
it in the form of Eq.~\eqref{eq:Fisher_renewal_splitting}, with 
\begin{subequations}\label{eq:renewal_two_fishers_single_expo}
\begin{align}
    F_{1:N}^{\rm ch}(\theta) &= N \sum_{k,q\in\obsSet}p(k|q) p_q \big(\partial_\theta \ln p(k|q)\big)^2\,,
    \\
    F_{1:N}^{\rm times|ch}(\theta) &= N \sum_{k,q\in\obsSet}p(k|q) p_q \big(\partial_\theta \ln b_{kq}\big)^2\,.
\end{align}
\end{subequations}
The Fisher information contained in the channels is therefore the same as in the general renewal case of Eq.~\eqref{eq:Fisher_renewal_channels}, but the second contribution now simplifies considerably. 

We now discuss physical scenarios in which the single-exponential assumption of Eq.~\eqref{eq:renewal_single_exponential} will appear. 
Recall that for renewal processes $L_k = \sqrt{\gamma_k} \ket{\mu_k}\bra{\nu_k}$, where  $\ket{\nu_k}$ and $\ket{\mu_k}$ are generic (normalized) kets; i.e., they in principle do not need to be elements of the same basis.
A WTD with a single exponential will occur  (i)~$\ket{\nu_k}$ are orthogonal among each other and (ii)~satisfy 
$[H,\ket{\nu_k}\bra{\nu_k}]=0$ for all $k$.
The $\ket{\mu_k}$ can still be arbitrary.
This therefore means that the post-jump states $\sigma_k= \ket{\mu_k}\bra{\mu_k}$ can have arbitrary coherences, but the pre-jump states must be diagonal in the eigenbasis of $H$. 
Referring to Eq.~\eqref{eq:renewal_prob_amplitude}, this  implies that 
\begin{equation}
    \Psi(\tau,k|q) = \sqrt{\gamma_k} e^{-i H_{kk} t} e^{-\gamma_k t/2} \braket{\nu_k | \mu_k}\,,
\end{equation}
where $H_{kk} = \braket{ \nu_k | H |\nu_k }$. It then follows that 
\begin{equation}\label{eq:renewal_WTD_single_exponential}
    W(\tau,k|q) = \gamma_k e^{-\gamma_k t}| \braket{ \nu_k | \mu_q } |^2\,,
\end{equation}
which has exactly the form in Eq.~\eqref{eq:renewal_single_exponential}, with transition probability between jumps $p(k|q) = |\braket{ \nu_k | \mu_q }|^2$ and decay rate $b_{kq} = \gamma_k$.
Interestingly, in this case $b_{kq}$ only depends on the post-jump state, so that Eq.~\eqref{eq:renewal_splitting_waiting_times_conditioning} simplifies to
\begin{equation}
    W(\tau|k,q) \equiv W(\tau|k) = \gamma_k e^{-\gamma_k t}\,.
\end{equation}
That is, the WTD between jumps in channels $q\to k$ depends only on the final jump state $k$.

In practice, the above scenario appears most often when dealing with incoherent systems described by classical (Pauli) master equations. 
Let $H\ket{i} = E_i \ket{i}$. Pauli equations occur when the jump operators have the form 
$L_{ji} = \sqrt{R_{ji}} \ket{j}\bra{i}$, where $R_{ji}$ is the transition rate to go from $\ket{i} \to \ket{j}$ (with $i\neq j$). 
Because the transitions occur between energy eigenstates, the long-time dynamics will be incoherent in the basis $\ket{i}$. 
In fact, starting from the original master Eq.~\eqref{eq:GKSL}, one may show that the probabilities $p_i = \braket{ i |\rho |i }$ will satisfy the Pauli rate equation 
\begin{equation}\label{eq:renewal_classical_ME}
    \frac{dp_i}{dt} = \sum_j R_{ij} p_j - \Gamma_i p_i,\qquad \Gamma_i = \sum_j R_{ji}\,.
\end{equation}
In quantum master equations we label the jumps by the channels $k$ appearing in the jump operators $L_k$. 
For classical master equations the index $k$ is a composite index $k \equiv i\to j$, representing the initial and final state of that jump operator, $L_{ji} = \sqrt{R_{ji}} \ket{j}\bra{i}$. 
If all channels are monitored, a jump $q = n\to m$ can  be followed by a jump $k = i\to j$ if and only if $m \equiv i$. 
The corresponding WTD is readily found from Eq.~\eqref{eq:renewal_W_as_probability_amplitude} or [\eqref{eq:renewal_WTD_single_exponential}], and reads 
\begin{equation}
    W(\tau, i\to j| n\to m) = \delta_{m,i} R_{ji} e^{-\Gamma_i \tau}\,.
\end{equation}
This is clearly in the form of Eq.~\eqref{eq:renewal_single_exponential}, so the results in Eq.~\eqref{eq:renewal_two_fishers_single_expo} apply.
However, some bookkeeping is required to translate the ``$q\to k$'' notation to ``$n\to m = i \to j$.''
The transition probability $p(k|q)$ in Eq.~\eqref{eq:renewal_markov_transition} becomes 
$p(i\to j| n\to m) = \delta_{m,i} \frac{R_{ji}}{\Gamma_i}$, while $b_{kq}$ is replaced with 
$b_{i\to j, n\to m} = \Gamma_i$.

We also need $p_q \equiv p_{n\to m}$, which is the solution of Eq.~\eqref{eq:renewal_markov_chain_jump_channels}:
\begin{equation}
    \sum_{n} \frac{R_{ji}}{\Gamma_i} p_{n\to i} = p_{i\to j}\,.
\end{equation}
As one may verify, the solution reads 
\begin{equation}
    p_{n\to m} = \frac{R_{mn} p_n^{\rm ss}}{A}\,,
\end{equation}
where $A = \sum_i \Gamma_i p_i^{\rm ss}$ is the steady-state dynamical activity (number of jumps per unit time) and $p_i^{\rm ss}$ is the steady-state of the master equation~\eqref{eq:renewal_classical_ME}; i.e., the solution of $\sum_j R_{ij} p_j^{\rm ss} = \Gamma_i p_i^{\rm ss}$.
Physically, $p_{n\to m}$ is the probability of observing pairs of states $(n,m)$ in the jump sequence. 
Plugging this in Eq.~\eqref{eq:renewal_two_fishers_single_expo} then finally yields, after some simplifications (since $\Gamma_i = \sum_j R_{ji}$):
\begin{subequations}
\begin{align}
    F_{1:N}(\theta) &= \frac{N}{A} \sum_{i,j} p_i^{\rm ss} 
    \frac{(\partial_\theta R_{ji})^2}{R_{ji}}\,,
    \\
    F_{1:N}^{\rm ch}(\theta) &= \frac{N}{A} \sum_{i,j} R_{ji}p_i^{\rm ss} \big(\partial_\theta \ln R_{ji}/\Gamma_i\big)^2\,, \\
    F_{1:N}^{\rm times|ch}(\theta) &=  \frac{N}{A} \sum_{i}p_i^{\rm ss}  \frac{(\partial_\theta \Gamma_i)^2}{\Gamma_i}\,,
\end{align}
\end{subequations}
and $F_{1:N}^{\rm ch}(\theta) = F_{1:N}(\theta) - F_{1:N}^{\rm times|ch}(\theta)$.

\section{Bound on Fisher information}\label{app:FI_bound}
Even when Fisher information is difficult to calculate, it can be easier to calculate some bounds on the Fisher information. For example, if the theory is described by a probability distribution of a real variable $\textrm{P}(x)dx$ with mean $\mu$ and variance $\sigma^2$, then $F\geq \frac{(\partial_\theta \mu)^2}{\sigma^2}$. It is often the case that the mean and variance are relatively easy to calculate. In this appendix, a similar bound is derived in the case of WTDs of renewal processes.

The bound arises from the Cauchy-Schwarz inequality,
\begin{equation}
    \int d\tau\, f^2(\tau) p(\tau) \times \int d\tau\, g^2(\tau) p(\tau) \leq \left[ \int d\tau\, f(\tau)g(\tau) p(\tau) \right]^2\,,
\end{equation}
where $f$ and $g$ are any real functions, while $p$ is a non-negative function. For our purposes, the integral is evaluated over $[0,\infty)$, representing the time intervals. Observe that this inequality follows from the functional inner product, $\left\langle f,g \right\rangle = \int d\tau\, f(\tau)g(\tau) p(\tau)$. Here, 
\begin{subequations}
    \begin{align}
        f(\tau) &= \partial_\theta \log p(x)\,, \\
        p(\tau) &= W(\tau,k\vert q)\,.
    \end{align}
\end{subequations}
Evaluating the terms in the Cauchy-Schwarz inequality, first,
\begin{equation}\label{eq:cauchy_schwarz}
    \int d\tau\, f^2(\tau) p(\tau) = \int d\tau\, \frac{\left[\partial_\theta W(\tau,k\vert q)\right]^2}{W(\tau,k\vert q)}\,.
\end{equation}
Then
\begin{align*}
    \int d\tau\, f(\tau)g(\tau) p(\tau) &= \int d\tau\, g(\tau) \partial_\theta W(\tau,k\vert q) \\ 
    &= \int d\tau \partial_\theta\left\{ g(\tau) W(\tau,k\vert q) \right\} - d\tau\, [\partial_\theta g(\tau)] W(\tau,k\vert q) \\
    &= \partial_\theta\left\{ \mathbb{E}\left[ g(\tau) \middle\vert k,q \right] p\left( k\middle\vert q \right) \right\} - \mathbb{E}\left[ \partial_\theta g(\tau) \middle\vert k,q \right] p\left( k\middle\vert q \right)\,.
\end{align*}
The above expression can be simplified by demanding that the (conditional) expectation of $g(\tau)$ vanish. For a function $\tilde{g}(\tau)$ that does not satisfy this condition, one can shift the function:
\begin{equation}
    \tilde{g}(\tau) \to g(\tau) = \tilde{g}(\tau) - \mathbb{E}\left[ \tilde{g}(\tau) \middle\vert k,q \right]\,.
\end{equation}
So
\begin{equation}
    \int d\tau\, f(\tau)g(\tau) p(\tau) = - \mathbb{E}\left[ \partial_\theta g(\tau) \middle\vert k,q \right] p\left( k\middle\vert q \right)\,.
\end{equation}
Finally,
\begin{equation}
    \int d\tau\, g^2(\tau) p(\tau) = \var\left[ g(\tau) \middle\vert k,q \right] p\left( k\middle\vert q \right)\,,
\end{equation}
where $\mathbb{E}\left[ g^2(\tau) \middle\vert k,q \right]=\var\left[ g(\tau) \middle\vert k,q \right]$ because $\mathbb{E}\left[ g(\tau) \middle\vert k,q \right]=0$.

The resulting inequality Eq.~\eqref{eq:cauchy_schwarz} can be written as
\begin{equation}\label{eq:FI_bound_incomplete}
\int d\tau\, \frac{\left[\partial_\theta W(\tau,k\vert q)\right]^2}{W(\tau,k\vert q)} \geq \frac{\mathbb{E}\left[ \partial_\theta g(\tau) \middle\vert k,q \right]^2}{\var\left[ g(\tau) \middle\vert k,q \right]}p\left( k\middle\vert q \right) \,.
\end{equation}
Performing a weighted sum of the above inequality over $k$ and $q$ with weights $p_q$, yields the lower bound on Fisher information Eq.~\eqref{eq:Fisher_renewal},
\begin{equation}\label{eq:FI_bound_general0}
    F_{1:N}(\theta)/N \geq \sum_{k,q} \frac{\mathbb{E}\left[ \partial_\theta g(\tau) \middle\vert k,q \right]^2}{\var\left[ g(\tau) \middle\vert k,q \right]} p\left( k\middle\vert q \right) p_q\,,
\end{equation}
where $g(\tau)$ is any function of $\tau$ with $\mathbb{E}\left[ g(\tau) \middle\vert k,q \right]=0$.

One may also consider the Fisher information conditioned on the knowledge of the jumps. The derivation of a lower bound follows as above in Eq.~\eqref{eq:FI_bound_incomplete} by considering $k$ to be trivial (so, e.g., $p\left( k\middle\vert q \right)\to 1$) and $q\to(k,q)$,
\[
\int d\tau\, \frac{\left[\partial_\theta W(\tau\vert k, q)\right]^2}{W(\tau\vert k, q)} \geq \frac{\mathbb{E}\left[ \partial_\theta g(\tau) \middle\vert k,q \right]^2}{\var\left[ g(\tau) \middle\vert k,q \right]} \,.
\]
Comparing to Eq.~\eqref{eq:Fisher_renewal_times_given_channels} and taking the weighted sum over $k$ and $q$ with weights $p\left( k\middle\vert q \right) p_q$ yields a bound,
\begin{equation}\label{eq:FI_bound_general_cond}
    F_{1:N}^{\rm times|ch}(\theta)/N \geq \sum_{k,q} \frac{\mathbb{E}\left[ \partial_\theta g(\tau) \middle\vert k,q \right]^2}{\var\left[ g(\tau) \middle\vert k,q \right]} p\left( k\middle\vert q \right) p_q\,.
\end{equation}
This is precisely the same as Eq.~\eqref{eq:FI_bound_general0}, except that the bound is on the conditional part of the Fisher information, alone. Because the Fisher information can be decomposed as $F=F^{\rm ch} + F^{\rm times|ch}$, we can thus improve the bound in Eq.~\eqref{eq:FI_bound_general0} by adding the information contained in the observed channels,
\begin{equation}\label{eq:FI_bound_general}
    F_{1:N}(\theta) \geq F_{1:N}^{\rm ch}(\theta) + N \sum_{k,q} \frac{\mathbb{E}\left[ \partial_\theta g(\tau) \middle\vert k,q \right]^2}{\var\left[ g(\tau) \middle\vert k,q \right]} p\left( k\middle\vert q \right) p_q\,,
\end{equation}

The common form of Eq.~\eqref{eq:FI_bound_general_cond} uses $g(\tau)=\tau-\mu_{kq}$ with $\mu_{kq}=\mathbb{E}\left[ \tau \middle\vert k,q \right]$. Similarly denoting the variance $\sigma_{kq}^2=\var\left[ \tau \middle\vert k,q \right]$, a bound on Fisher information is obtained,
\begin{equation}\label{eq:FI_bound_avg}
     F_{1:N}^{\rm times|ch}(\theta) \geq N \sum_{k,q} \frac{\left[\partial_\theta \mu_{kq}\right]^2}{\sigma_{kq}^2} p\left( k\middle\vert q \right) p_q\,.
\end{equation}
This bound roughly arises from considering the average delay between observations.

One may consider how the calculations change when rescaling $g(\tau)$:
\begin{equation}
    \bar{g}(\tau) = h(\theta,k,q) g(\tau)\,.
\end{equation}
It remains the case that $\mathbb{E}\left[ \bar{g}(\tau) \middle\vert k,q \right] = h(\theta,k,q) \mathbb{E}\left[ g(\tau) \middle\vert k,q \right]=0$.
The variance is 
\[
\var\left[ \bar{g}(\tau) \middle\vert k,q \right] = h^2(\theta,k,q) \var\left[ g(\tau) \middle\vert k,q \right]\,.
\]
Then 
\begin{align*}
    \mathbb{E}\left[ \partial_\theta \bar{g}(\tau) \middle\vert k,q \right] &= h(\theta,k,q) \mathbb{E}\left[ \partial_\theta g(\tau) \middle\vert k,q \right] + \left( \partial_\theta h(\theta,k,q) \right) \mathbb{E}\left[ g(\tau) \middle\vert k,q \right] \\
    &= h(\theta,k,q) \mathbb{E}\left[ \partial_\theta g(\tau) \middle\vert k,q \right]\,.
\end{align*}
The factor of $h^2(\theta,k,q)$ from this expectation value (after squaring) cancels the same factor from the variance. That is, rescaling $g(\tau)$ does not affect the bound on Fisher information. Considering a bound based on a weighted average instead of the average, as in Eq.~\eqref{eq:FI_bound_avg}, will then result in the same lower bound. 

Though rescaling has no effect, the bound can be improved by considering the function $g(\tau)$ with higher order terms in $\tau$. For example, if one considers quadratic functions in $\tau$, a stronger bound is obtained in terms of the first four statistical moments~\cite{bhattacharyya_analogues_1946, jarrett_bounds_1984, stein_pessimistic_2017}.

\section{Monitoring operator for multi-parameter estimation}
\label{sect:Monitoring_operator_for_multi_parameter_estimation}
In this appendix, we present a generalization of the monitoring operator formalism to multi-parameter estimation. The parameters to be estimated are represented by the vector $\vec{\theta}$. The Fisher information matrix~\cite{Zegers_2015, Ly_2017} is defined as 
\begin{equation}
    \left[F(\vec{\theta})\right]_{ij} = \mathbb{E}\left[\frac{1}{\prob(X)^2}\partial_{\theta_i} \prob(X)\partial_{\theta_j} \prob(X)\right]\,.
\end{equation}
This expression can be rewritten in terms of the monitoring operator, introduced in Eq.~\eqref{eq:monitoring_operator}, by defining parameter-specific operators $\xi^{(i)}_t$, $\xi^{(j)}_j$, with the important property
\begin{equation}
    \tr\left[\xi_t^{(i)}\right] = \frac{1}{\prob(X)}\partial_{\theta_i} \prob(X)\,.
\end{equation}
The operators $\xi^{(i)}_t$ evolve in the same way as the operator $\xi_t$, but all the derivatives have to be taken with respect to $\theta_i$, while the other parameters are assumed to be fixed. In this way, we finally get the expression for the Fisher information matrix
\begin{equation}
    \left[F(\vec{\theta})\right]_{ij} = \mathbb{E}\left[\tr\left[\xi_t^{(i)}\right] \tr\left[\xi_t^{(j)}\right]\right]\,. 
    \label{eq:monitoring_multi_parameter}
\end{equation}
In order to obtain the Fisher information matrix, one should then separately evolve monitoring operators $\xi_t^{(i)}$ for each component of the parameters vector, and combine them according to Eq.~\eqref{eq:monitoring_multi_parameter}.

\section{Fisher information rate: average}
\label{sect:FI_rate_average}
In this appendix, we show the derivation of Eq.~\eqref{eq:Fisher_information_rate_average}, yielding the average Fisher information rate for a generic measurement record. The Fisher information variation is given by
\begin{equation}
    dF = F(X_{T+1}|X_{1:T}) = \mathbb{E}_{1:T}\left[\frac{1}{\prob(x_{T+1}|x_{1:T})}\left(\partial_\theta \prob(x_{T+1}|x_{1:T})\right)^2\right]\,,
\end{equation}
where the expectation value is taken with respect to all the possible trajectories of $T$ time steps, not including timestep $T+1$. We now introduce the stochastic currents $I_T^k$, such that $\prob(k|x_{1:t}) = I_t^k dt$, for $k$ labeling one of the $|\obsSet|$ monitored jump channels:
\begin{equation}
    dF = \mathbb{E}\left[\sum_{k=1}^{|\obsSet|} \frac{1}{I_T^k dt}\left(\partial_\theta I^k_T dt\right)^2 + \frac{1}{1-\sum_{k=1}^{|\obsSet|} I_T^k dt}\left(\sum_{k=1}^{|\obsSet|} \partial_\theta I_T^k dt\right)^2\right]\,.
\end{equation}
We note that the first term, which only contains single-current contributions, is of order $dt$ in the $dt\to 0$ limit. On the contrary, the second term, which contains also some non-diagonal contributions mixing currents together, is of order $dt^2$. To compute the rate, only terms of order $dt$ should be taken into account; we obtain therefore
\begin{equation}
    dF = \mathbb{E} \left[ \sum_{k=1}^{|\obsSet|} \frac{1}{I_T^k dt}\left(\partial_\theta I_T^k dt\right)^2 \right]\,,
\end{equation}
from which the Fisher information rate follows.

\end{widetext}

\begin{algorithm*}[htbp]
\caption{Gillespie evolution for the stochastic Fisher information}\label{alg:Gillespie}
\KwIn{Hamiltonian $H_\theta$; displaced Hamiltonians $H_{\theta+d\theta}$ and $H_{\theta-d\theta}$;}
list of monitored jumps Kraus operators [$M_\theta$]\;
displaced lists of monitored jumps Kraus operators [$M_{\theta+d\theta}$] and [$M_{\theta-d\theta}$]\;
derivation increment $d\theta$; initial state $\ket{\psi_0}$; final time $t_f$; timestep $dt$; number of trajectories\;

\Precompute\
$J=\sum_k M_k^\dagger M_k$\;
Effective Hamiltonian $H_e = H - \frac{i}{2}J$; Displaced effective Hamiltonian $H_{e\pm} = H_{\theta\pm d\theta} - \frac{i}{2} J_{\pm}$\;
Displaced jump operators $J_\pm = \sum_k M_{k,\theta\pm d\theta}^\dagger M_{k,\theta\pm d\theta}$\;
Let $V[t]$ be the list of the no-jump evolution operators\;
Let $Q_s[t]$ be the list of the non-state-dependent parts of the WTD\;
Let $\dot{V}[t]$ be the list of the derivatives of the no-jump evolution operators with respect to $\theta$\;
\For{all times $t < t_f$}{
    $V[t] = e^{-i H_e t}$\;
    $Q_s[t] = V[t]^\dagger J V[t]$\;
    $\dot{V}[t] = \left[\exp(-i H_{e+} t) - \exp(-i H_{e-} t)\right] / 2d\theta$\;
}
Let [$\dot{M}$] be the list of the derivatives of the jump operators with respect to the parameter $\theta$\;
Let $\Delta[t]$ be the list of all the derivatives of $M_kV[t]$ for all values of $k$\;
\For{all the jump operators in $[M]$}{
    $\dot{M}_k = \left(M_{k+} - M_{k-}\right) / 2d\theta$\;
    Let $\Delta_k$ be the list of the derivatives of $M_kV[t]$\;
    \For{$t < t_f$}{
        $\Delta_k[t] = \dot{M}_k V[t] + M_k \dot{V}[t]$\;
    }
    $\Delta[k] = \Delta_k$\;
}

\For{all trajectories}{
    Fix initial state $\ket{\psi}=\ket{\psi_0}$, initial monitoring operator $\xi = 0$ and initial time $\tau = 0$\;
    \While{$t < t_f$}{
        $\rho = \ket{\psi}\bra{\psi}$\;
        Let $W[t]$ be the list to contain the WTD\;
        \For{every value $Q$ in $Q_s$}{
            $W[t] = \braket{\psi|Q[t]|\psi}$\;
        }
        Sample a time $T$ from the distribution $W[t]$\;
        $\tau += T$\;
        Update the state accordingly: $\ket{\psi} = V[T]\ket{\psi}$\;
        Let $[p]$ be the list of the jump probabilities\;
        \For{all the jump operators in $[M]$}{
            $p = \braket{\psi | M_k^\dagger M_k | \psi}$\;
        }
        Sample a jump channel $\bar{k}$ from the weights distribution $[p]$\;
        Update and normalize the state $\ket{\psi} = M_{\bar{k}} \ket{\psi} / ||\psi||$\;
        Compute the evolution of the $\xi$ operator.
        $\xi = \frac{1}{||\psi||^2} \left(M_{\bar{k}} V[T] \xi V[T]^\dagger M_{\bar{k}}^\dagger + \Delta[\bar{k}][T] \rho V[T]^\dagger M_{\bar{k}}^\dagger + M_{\bar{k}} V[T] \rho \Delta[\bar{k}][T]^\dagger\right)$\;
        Compute the stochastic Fisher information $f = \tr[\xi]^2$\;
    }
}
\KwOut{[$(t_j, k_j, \rho_j, f_j)$] vectors of measurement records, states and stochastic Fisher information values for all the trajectories, $V[t]$, $\dot{V}[t]$ often needed for further calculations, and the list of times at which $V$ is known.}
\end{algorithm*}

\end{document}